\documentclass[11pt]{article}
\usepackage[utf8]{inputenc}
\usepackage{amsfonts,latexsym,amsthm,amssymb,amsmath,amscd,euscript,mathrsfs,graphicx}
\usepackage{framed}
\usepackage{fullpage}
\usepackage{color}
\usepackage[obeyFinal,textsize=scriptsize,shadow]{todonotes}
\usepackage{tikz}
\usetikzlibrary{matrix}
\usepackage{hyperref}
\usepackage{cleveref}

\RequirePackage{algorithm}
\RequirePackage{algorithmic}

\usepackage{subfig}

\theoremstyle{plain}
\newtheorem{theorem}{Theorem}[section]
\newtheorem{proposition}[theorem]{Proposition}
\newtheorem{lemma}[theorem]{Lemma}

\theoremstyle{definition}
\newtheorem{definition}[theorem]{Definition}

\newtheorem{question}[theorem]{Question}
\theoremstyle{remark}

\newcommand{\BE}{\mathbb E}

\newcommand{\BZ}{\mathbb Z}
\newcommand{\eps}{\varepsilon}

\DeclareMathOperator{\MWM}{MWM}
\DeclareMathOperator{\MST}{MST}
\DeclareMathOperator{\PF}{PF}
\DeclareMathOperator{\diam}{diam}
\DeclareMathOperator{\HST}{HST}
\DeclareMathOperator{\OPT}{OPT}
\DeclareMathOperator{\cost}{cost}
\DeclareMathOperator{\GMM}{GMM}
\renewcommand{\div}{\mathrm{div}}

\title{Improved Diversity Maximization Algorithms for Matching and Pseudoforest}
\author{Sepideh Mahabadi\thanks{Microsoft Research, Redmond. \texttt{smahabadi@microsoft.com}} , Shyam Narayanan\thanks{Massachusetts Institute of Technology. Work done as an intern at Microsoft Research. \texttt{shyamsn@mit.edu}}}
\date{\today}

\begin{document}

\maketitle

\begin{abstract}
    In this work we consider the diversity maximization problem, where given a data set $X$ of $n$ elements, and a parameter $k$, the goal is to pick a subset of $X$ of size $k$ maximizing a certain diversity measure. 
    \cite{chandra2001diversitymax} defined a variety of diversity measures based on pairwise distances between the points.
    A constant factor approximation algorithm was known for all those diversity measures except ``remote-matching'', where only an $O(\log k)$ approximation was known. 
    In this work we present an $O(1)$ approximation for this remaining notion.
    Further, we consider these notions from the perpective of composable coresets. \cite{indyk2014diversity} provided composable coresets with a constant factor approximation for all but ``remote-pseudoforest'' and ``remote-matching'', which again they only obtained a $O(\log k)$ approximation. Here we also close the gap up to constants and present a constant factor composable coreset algorithm for these two notions.
    For remote-matching, our coreset has size only $O(k)$, and for remote-pseudoforest, our coreset has size $O(k^{1+\eps})$ for any $\eps > 0$, for an $O(1/\eps)$-approximate coreset.
\end{abstract}

\newpage

\pagenumbering{arabic}

\section{Introduction}
Diverse Subset Selection is the task of searching for a subset of the data that preserves its diversity as much as possible. This task has a large number of applications in particular while dealing with large amounts of data, including data summarization (e.g.~\cite{lin2009graph, gong2014diverse, lin2011class}), search and information retrieval (e.g.~\cite{angel2011efficient, abbar2013diverse, jain2004providing, drosou2010search, welch2011search}), and recommender systems  (e.g.~\cite{zhou2010solving, abbar2013real, ziegler2005improving, yu2009recommendation}), among many others (e.g. \cite{gollapudi2009axiomatic, pilourdault2017motivation}). 
Here, given a ground set of $n$ vectors $X$ in a metric space $(\mathcal{X},\rho)$, representing a data set of objects (e.g. using their feature vectors), and a parameter $k$, the goal is to choose a subset $S\subseteq X$ of this data set of size $k$, that maximizes a pre-specified optimization function measuring the diversity. 

Many diversity measures have been introduced and used in the literature that fit different tasks. A large number of these measures are defined based on pairwise distances between the vectors in $X$. In particular the influential work of \cite{chandra2001diversitymax} introduced a taxonomy of pairwise-distance based diversity measures which is shown in Table~\ref{tab:diversity}. For example, remote-edge measures the distance of the closest points picked in the subset $S$, and remote-clique measures the sum of pairwise distances between the points in the subset $S$. Remote-pseudoforest falls between these two where it wants to ensure that the average distance of a point to its nearest neighbor is large.  
Remote-matching measures the diversity as the cost of minimum-weight-matching.
Various other measures have also been considered:
Table \ref{tab:diversity} includes each of their definitions along with the best known approximation factor for these measures known up to date. In particular, by \cite{chandra2001diversitymax} it was known that all these measures except remote-pseudoforest and remote-matching admit a constant factor approximation. More recently, \cite{bhaskara2016pseudoforest} showed a constant factor randomized LP-based algorithm for remote-pseudoforest. They also showed the effectiveness of the remote-pseudoforest measure on real data over the other two common measures (remote-edge and remote-clique).
On the lower bound side, it was known by \cite{halldorsson1999finding} that for remote-matching, one cannot achieve an approximation factor better than $2$. However, despite the fact that there has been a large body of work on diversity maximization problems \cite{borodin2012max, abbassi2013diversity, ceccarello2016mapreduce, cevallos2018diversity,epasto2022improved},
the following question had remained unresolved for over two decades.
\begin{question}
Is it possible to get an $O(1)$ approximation algorithm for the remaining notion of remote-matching?
\end{question}

\begin{table*}
\begin{center}
 {\small
 \begin{tabular}[h]{|l|l|l|}
  \hline 
    Problem & Diversity of the point set $S$ & Apx factor \\   \hline 
    Remote-edge & $\min_{x,y\in S} dist(x,y)$ &  $O(1)$\\ 
    Remote-clique & $\sum_{x,y\in S} dist(x,y)$ & $O(1)$  \\ 
    Remote-tree& $wt(MST(S) )$, weight of the minimum spanning tree of $S$& $O(1)$ \\
    Remote-cycle & $\min_C wt(C)$ where $C$ is a TSP tour on $S$& $O(1)$\\
    Remote $t$-trees & $\min_{S=S_1|...|S_t} \sum_{i=1}^t wt(MST(S_i))$ & $O(1)$\\
    Remote $t$-cycles & $\min_{S=S_1|...|S_t} \sum_{i=1}^t wt(TSP(S_i))$ & $O(1)$\\
    Remote-star & $\min_{x\in S} \sum_{y\in S\setminus\{x\}} dist(x,y)$& $O(1)$\\
    Remote-bipartition & $\min_B wt(B)$, where $B$ is a bipartition (i.e., bisection) of $S$& $O(1)$\\
    Remote-pseudoforest & $\sum_{x\in S} \min_{y\in S\setminus\{x\}} dist(x,y)$& $O(1)$ (Offline)\\
    & & $O(\log k)$ (Coreset) \\
    Remote-matching & $\min_M wt(M)$, where $M$ is a perfect matching of $S$& $O(\log k)$\\
    \hline 
  \end{tabular}
  }
\end{center}
\caption{This table includes the notions of diversity considered by \cite{chandra2001diversitymax} ($S=S_1|...|S_t$ is used to denote that $S_1 \ldots S_t$ is a partition of $S$ into $t$ sets). We also include the best previously-known approximation factors, both in the standard (offline) and coreset setting. If not explicitly stated, the approximation factor holds for both the offline setting and the coreset setting. We note that the previously known $O(1)$-approximate remote-pseudoforest offline algorithm is randomized, whereas the rest of the previously known algorithms are deterministic.}
\label{tab:diversity}
\end{table*}

Later following a line of work on diversity maximization in big data models of computations, \cite{indyk2014diversity}  presented algorithms producing a composable coreset for the diversity maximization problem under all the aformentioned diversity measures. An $\alpha$-approximate composable coreset for a diversity objective is a mapping that given a data set $X$, outputs a small subset $C \subset X$ with the following composability property: given multiple data sets $X^{(1)},\cdots,X^{(m)}$, the maximum achievable diversity over the union of the composable coresets $\bigcup_i C^{(i)}$ is within an $\alpha$ factor of the maximum diversity over the union of those data sets $\bigcup_i X^{(i)}$. It is shown that composable coresets naturally lead to solutions in several massive data processing models including distributed and streaming models of computations, and this has lead to recent interest in composable coresets since their introduction \cite{mirrokni2015randomized, assadi2017randomized, mahabadi2019composable, indyk2020composable, aghamolaei2015diversity, epasto2019scalable}.
\cite{indyk2014diversity} showed $\alpha$-approximate compsable coresets again for all measures of diversity introduced by \cite{chandra2001diversitymax}. 
They presented constant factor $\alpha$-approximate composable coresets for all diversity measures except remote-pseudoforest and remote-matching, where they provided only $O(\log k)$-approximations for these measures. Again the following question remained open.
\begin{question}
Is it possible to get a constant factor composable coreset for the remote-pseudoforest and remote-matching objective functions? 
\end{question}

In this work we answer above two questions positively and close the gap up to constants for these problems.

\subsection{Our Results}

In this work, we resolve a longstanding open question of~\cite{chandra2001diversitymax} by providing polynomial-time $O(1)$-approximation algorithms for the remote-matching problem. We also resolve a main open question of~\cite{indyk2014diversity} by providing polynomial-time algorithms that generate $O(1)$-approximate composable coresets for both the remote-pseudoforest and remote-matching problems. Hence, our paper establishes $O(1)$-approximate offline algorithms \emph{and} $O(1)$-approximate composable coresets for \emph{all} remaining diversity measures proposed in~\cite{chandra2001diversitymax}.
Specifically, we have the following theorems.

\begin{theorem}[Remote-Matching, Offline Algorithm] \label{thm:matching_offline}
    Given a dataset $X = \{x_1, \dots, x_n\}$ and an integer $k \le \frac{n}{3}$, w.h.p. Algorithm \ref{alg:mwm_offline} outputs an $O(1)$-approximate set for remote-matching.
\end{theorem}

While we assume $n$ is at least a constant factor bigger than $k$, this is usually a standard assumption (for instance both the remote-pseudoforest and remote-matching algorithms in \cite{chandra2001diversitymax} assume $n \ge 2k$, and the $O(1)$-approximate remote-pseudoforest algorithm in \cite{bhaskara2016pseudoforest} assumes $n \ge 3k$). 
We now state our theorems regarding composable coresets.

\begin{theorem}[Pseudoforest, Composable Coreset] \label{thm:pseudoforest_coreset}
    Given a dataset $X = \{x_1, \dots, x_n\}$ and an integer $k \le n$, Algorithm \ref{alg:pf_coreset} outputs an $O(1/\eps)$-approximate composable coreset $C$ for remote-pseudoforest, of size at most $O(k^{1+\eps})$.
    
    By this, we mean that if we partitioned the dataset $X$ into $X^{(1)}, \dots, X^{(m)}$, and ran the algorithm separately on each piece to obtain $C^{(1)}, \dots, C^{(m)},$ each $C^{(i)}$ will have size at most $O(k^{1+\eps})$ and $\max_{Z \subset C: |Z| = k} \PF(Z) \ge \Omega(\eps) \cdot \max_{Z \subset X: |Z| = k} \PF(Z)$, where $C = \bigcup_{i=1}^m C^{(i)}$.
\end{theorem}

\begin{theorem}[Minimum-Weight Matching, Composable Coreset] \label{thm:matching_coreset}
    Given a dataset $X = \{x_1, \dots, x_n\}$ and an integer $k \le n$, Algorithm \ref{alg:mwm_coreset} outputs an $O(1)$-approximate composable coreset $C$ for remote-matching, of size at most $3k$.
\end{theorem}

In all of our results, we obtain $O(1)$-approximation algorithms, whereas the previous best algorithms for all $3$ problems was an $O(\log k)$-approximation algorithm, meaning the diversity was at most $\Omega(\frac{1}{\log k})$ times the optimum. We remark that our composable coreset in Theorem \ref{thm:pseudoforest_coreset} has size $k^{1+\eps},$ for some arbitrarily small constant $\eps$ (to obtain a constant approximation) which is possibly suboptimal. We hence ask an open question as to whether an $O(1)$-approximate composable coreset for remote-pseudoforest, of size $O(k)$, exists. 

Finally, as an additional result we also prove an alternative method of obtaining an $O(1)$-approximate offline algorithm for remote-pseudoforest. Unlike~\cite{bhaskara2016pseudoforest}, which uses primal-dual relaxation techniques, our techniques are much simpler and are based on $\eps$-nets and dynamic programming. In addition, our result works for all $k \le n$, whereas the work of~\cite{bhaskara2016pseudoforest} assumes $k \le \frac{n}{3}$ and that $k$ is at least a sufficiently large constant. Also, our algorithm is deterministic, unlike~\cite{bhaskara2016pseudoforest}. We defer the statement and proof to \Cref{sec:offline_pseudoforest}.

\section{Preliminaries}

\subsection{Definitions and Notation}

We use $\rho(x, y)$ to represent the metric distance between two points $x$ and $y$. For a point $x$ and a set $S$, we define $\rho(x, S) = \min_{s \in S} \rho(x, s)$. Likewise, for two sets $S, T$, we define $\rho(S, T) = \min_{s \in S, t \in T} \rho(S, T)$.
We define the diameter of a dataset $X$ as $\diam(X) = \max_{x, y \in X} \rho(x, y).$

\paragraph{Costs and diversity measures.} For a set of points $Y = \{y_1, \dots, y_k\}$, we use $\div(Y)$, as a generic term to denote its diversity which we measure by the following cost functions.

If $k = |Y|$ is even, we define the minimum-weight matching cost $\MWM(Y)$ as the minimum total weight over all perfect matchings of $Y$. Equivalently,
\[\MWM(Y) := \min\limits_{\text{permutation } \pi: [k] \to [k]} \sum_{i=1}^{k/2} \rho(y_{\pi(2i)}, y_{\pi(2i-1)}).\]
Likewise, we define the pseudoforest cost $\PF(Y)$, also known as the sum-of-nearest-neighbor cost, of $Y$ as
\[\PF(Y) := \sum_{y \in Y} \rho(y, Y \backslash y).\]
Finally, we define the minimum spanning tree cost $\MST(Y)$ of $Y$ as the minimum total weight over all spanning trees of $Y$. Equivalently,
\[\MST(Y) = \min\limits_{G: \text{ spanning tree of } [k]} \sum_{e = (i, j) \in G} \rho(y_i, y_j).\]
The pseudoforest cost and minimum spanning tree cost do not require $Y$ to be even.

\begin{definition}[Diversity maximization] Given a dataset $X$, and a parameter $k$, the goal of the diversity maximization problem is to choose a subset $Y\subset X$ of size $k$ with maximum diversity, where in this work we focus on $\div(Y)=\MWM(Y)$ and $\div(Y)=\PF(Y)$.

We also define $\div_k(X)$ to be this maximum achievable diversity, i.e., $\div_k(X)=\max_{Y\subset X, |Y|=k}\div(Y)$. In particular we use $\MWM_k(X)$, or the \emph{remote-matching} cost of $X$, as $\max_{Y\subset X, |Y|=k}\MWM(Y)$, and define $\PF_k(X)=\max_{Y\subset X, |Y|=k}\PF(Y)$. These objectives are also known as \emph{$k$-matching} and \emph{$k$-pseudoforest}.
\end{definition}

For a specific diversity maximization objective $\div_k$ (such as $\MWM_k$), an $\alpha$-\emph{approximation algorithm} ($\alpha \ge 1$) for $\div$ is an algorithm that, given any dataset $X = \{x_1, \dots, x_n\}$, outputs some dataset $Z \subset X$ of size $k$ such that 
\[\div(Z) \ge \frac{1}{\alpha} \cdot \div_k(X) = \frac{1}{\alpha} \cdot \max_{Y \subset X: |Y| = k} \div(Y).\]

\begin{definition}[Composable coresets]We say that an algorithm $\mathcal{A}$ that acts on a dataset $X$ and outputs a subset $\mathcal{A}(X) \subset X$ forms an $\alpha$-\emph{approximate composable coreset} ($\alpha \ge 1$) for $\div$ if, for any collection of datasets $X^{(1)}, \dots, X^{(m)},$
we have
\[\div_k\left(\bigcup_{i=1}^m \mathcal{A}(X^{(i)})\right) \ge \frac{1}{\alpha} \cdot \div_k\left(\bigcup_{i=1}^m X^{(i)}\right).\]
\end{definition}

Throughout the paper, for our coreset construction algorithms, we use $X^{(1)},\cdots,X^{(m)}$ to denote the collection of the data sets.
Note that since $\mathcal{A}(X^{(i)}) \subset X^{(i)},$ the diversity of the combined coresets is always at most the diversity of the combined original datasets. We also say that the coreset is of size $k'$ if $|\mathcal{A}(X^{(i)})| \le k'$ for each $X^{(i)}$. We desire for the size $k'$ to only depend (polynomially) on $k$, and not on $n$.

\subsection{The GMM Algorithm}

The GMM algorithm~\cite{gonzalez1985clustering} is an algorithm that was first developed for the $k$-center clustering, but has since been of great use in various diversity maximization algorithms and dispersion problems, starting with~\cite{ravi1991dispersion}. The algorithm is a simple greedy procedure that finds $k$ points $Y$ in a dataset $X$ that are well spread out. It starts by picking an arbitrary point $y_1 \in X$. Given $y_1, \dots, y_p$ for $p < k,$ it chooses $y_{p+1}$ as the point that maximizes the distance $\rho(y, \{y_1, \dots, y_p\})$ over all choices of $y \in X$. We provide pseudocode in Algorithm \ref{alg:gmm}.

\begin{algorithm}[tb]
   \caption{The GMM Algorithm}
   \label{alg:gmm}
\begin{algorithmic}[1]
    \STATE {\bfseries Input:} data $X= \{x_1, \dots, x_n\}$, integer $k$.
    \STATE $y_1$ is an arbitrary point in $X$.
    \STATE \textbf{Initialize} $Y \leftarrow \{y_1\}$.
    \FOR{$p = 2$ to $k$}
        \STATE $y_p \leftarrow \arg\max_{y \in X} \rho(y, Y)$, i.e., $y_p$ is the furthest point in $X$ from the current $Y = \{y_1, \dots, y_{p-1}\}.$
        \STATE $Y \leftarrow Y \cup \{y_p\}$.
    \ENDFOR
    \STATE \textbf{Return} $Y$.
\end{algorithmic}
\end{algorithm}

The GMM algorithm serves as an important starting point in many of our algorithms, as well as in many of the previous state-of-the-art algorithms for diversity maximization.
The GMM algorithm has the following crucial property.

\begin{proposition} \label{prop:gmm}
    Suppose we run GMM for $k$ steps to produce $Y = \{y_1, \dots, y_k\}$. Let $r = \max_{x \in X} \rho(x, Y)$. Then, every pair of points $y_i, y_j$ has $\rho(y_i, y_j) \ge r$.
\end{proposition}

\section{Composable Coreset Constructions} \label{sec:composable_coreset}

In this section, we design algorithms for constructing $O(1)$-approximate composable coresets for remote-pseudoforest and remote-matching. In this section, we provide a technical overview and pseudocode for the algorithms, but defer the proof (and algorithm descriptions in words) to \Cref{sec:coreset_pseudoforest} (for remote-pseudoforest) and \Cref{sec:coreset_matching} (for remote-matching).

\subsection{Coreset Constructions: Technical Overview} \label{subsec:coreset_overview}

For both remote-pseudoforest and remote-matching, we start by considering the heuristic of GMM, where in each group we greedily select $k$ points. For simplicity, suppose we only have one group for now. After picking the set $Y = \{y_1, \dots, y_k\}$ from GMM, define $r$ to be the maximum distance $\rho(x, Y)$ over all remaining points $x$. Then, Proposition \ref{prop:gmm} implies that all distances $\rho(y_j, y_{j'}) \ge r$ for $j, j' \le k$. Hence, running GMM will ensure we have a set of $k$ points with minimum-weight-matching or pseudoforest cost at least $\Omega(k \cdot r).$ Hence, we only fail to get an $O(1)$-approximation if the optimum remote-matching (or remote-pseudoforest) cost is much larger than $k \cdot r$.

However, in this case, note that every point $x \in X$ satisfies $\rho(x, Y) \le r,$ meaning every point $x$ is within $r$ of some $y_i$. This suggests that if the optimum cost is $\omega(k \cdot r)$, achieved by some points $z_1, \dots, z_k,$ we could just map each $z_i$ to its closest $y_i$, and this would change each distance by no more than $O(k \cdot r)$. Hence, we can ostensibly use the GMM points to obtain a cost within $O(k \cdot r)$ of the right answer, which is within an $\Omega(1)$ (in fact a $1-o(1)$) multiplicative factor! Additionally, this procedure will compose nicely, because if we split the data into $m$ components $X^{(1)}, \dots, X^{(m)},$ each with corresponding radius $r^{(1)}, \dots, r^{(m)}$, then each individual coreset has cost at least $r^{(j)}$ (so the combination has cost at least $\max r^{(i)}$), whereas we never move a point more than $r^{(i)} \le \max r^{(i)}$.

The problem with this, however, is that we may be using each $y_j$ multiple times: for instance, if both $z_1$ and $z_2$ are closest to $y_1,$ we would use $y_1$ twice. Our goal is to find a subset of $k$ points, meaning we cannot duplicate any point.

Note that in this duplication, it is never necessary to duplicate a point more than $k$ times. So, if we could somehow pick $k$ copies of each $y_i$, we would have a coreset. However, note that it is not crucial for each $z_i$ to be mapped to its closest GMM point $y_j$: any point within distance $r$ of $y_j$ is also acceptable. Using this observation, it suffices to pick $k$ points among those closest to $y_j$ if possible - if there are fewer than $k$ points, picking all of the points is sufficient. It is also important to choose all of $Y$, in the case where the optimum cost is only $O(k \cdot r)$. Together, this generates a composable coreset for both remote-matching and remote-pseudoforest, of size only $k^2$.

\paragraph{Improving the coreset for remote-matching.} In the case of remote-matching, we can actually improve this to $O(k)$. The main observation here is to show that if a set $Z$ of size $k$ had two identical points, getting rid of both of them does not affect the minimum-weight-matching cost. (This observation does \emph{not} hold for pseudoforest). One can similarly show that if the two points were close in distance, removing both of the points does not affect the matching cost significantly. This also implies we can, rather than removing both points, move them both to a new location as long as they are close together. At a high level, this means that there must exist a near-optimal $k$-matching that only has $O(1)$ points closest to each $y_j$: as a result we do not have to store $k$ points for each $y_j$: only $O(1)$ points suffice. 

\paragraph{Improving the coreset for remote-pseudoforest.} In the case of remote-pseudoforest, the improvement is more involved. Consider a single group, and suppose GMM gives us the set $Y = \{y_1, \dots, y_k\}$. Let $X_i$ represent the set of points in $X$ closest to $Y_i$. The first observation we make is that if the optimal solution had multiple points in a single $X_i$, each such point can only contribute $O(r)$ cost. Assuming that the optimum cost is $\omega(k \cdot r)$, it may seem sufficient to simply pick $1$ point in each $Y_i$ for the coreset, as we can modify the optimum solution by removing points in $X_i$ if there are two or more of them. While this will allow us to obtain a set with nearly optimum cost, the problem is the set has size less than $k$. So, we need to add additional points while preventing the pseudoforest cost from decreasing by too much.

To develop intuition for how this can be accomplished, we first suppose that $|X_1|, |X_2| \ge k$. In this case, we could choose the coreset as $X_1 \cup X_2 \cup Y$. We know there is a subset $Z \subset Y$ with pseudoforest cost close to optimum, though $|Z|$ may be much smaller than $k$. However, since $y_1, y_2$ are far away from each other (they were chosen first in the greedy procedure of GMM), so all points in $X_1$ and all points in $X_2$ are far from each other. That means that if we randomly choose either $X_1$ or $X_2$, and add enough points from the chosen set so that we have $k$ points, each point $y \in Y$, with $50\%$ probability, is not close to the new points added (because every point $y$ must either be far from $X_1$ or far from $X_2$). Thus, the expected distance from $y$ to the closest new point is large.

In general, it is not actually important that the points come from $X_1$ and $X_2$: we just need two sets $S, T$ of $k$ points such that $\rho(S, T)$, the minimum distance between $s \in S$ and $t \in T$, is large. Then, any point $y$ cannot be close to points in both $S$ and $T$. This also composes nicely, because to find the final set of $k$ points, we only need there to be two sets $S, T$ throughout the union of the coresets with large $\rho(S, T)$.

To find large $S, T$ with large $\rho(S, T)$, we will require $|X| \ge k^{1+\eps}$ for some small constant $\eps$. For simplicity, we focus on the case when $|X| \ge k^{1.5}$. Suppose all but $k$ points are in some ball $B$ of radius $r$. If there exists $x$ that is within distance $r/10$ of $k$ points (we can make $S$ these $k$ points), then all points in $S$ must be far away from the furthest $k$ points from $x$ (which we can set as $T$), or else we could have found a smaller ball $B'$.
Otherwise, there are two options.
\begin{enumerate}
    \item The majority of points $x \in X$ are within $r/100$ of at least $\sqrt{k}$ other points, but no $x \in X$ is within $r/10$ of at least $k$ other points. Intuitively (we will make this intuition formal in \Cref{sec:coreset_pseudoforest}), a random set $S_0$ of size $O(\sqrt{k})$ should be within $r/100$ of at least $k$ other points in total (we can make $S$ these $k$ points), but there are at least $|X| - k \cdot |S_0| \ge k$ points (we can make $T$ these points) that are not within $r/10$ of $S_0$.
    \item The majority of points aren't within $r/100$ of even $O(\sqrt{k})$ points. In this case, we can pick $k$ of these points to form $S$, and they will not be within $r/100$ of at least $|X| - O(\sqrt{k}) \cdot |S| \ge k$ points. We make this intuition formal and prove the result for the more general $k^{1+\eps}$.
\end{enumerate}

\subsection{Algorithm Pseudocode}

We provide pseudocode for the remote-pseudoforest coreset in Algorithm \ref{alg:pf_coreset} and for the remote-matching coreset in Algorithm \ref{alg:mwm_coreset}. The proofs, as well as algorithm descriptions in words, are deferred to \Cref{sec:coreset_pseudoforest} (for remote-pseudoforest) and \Cref{sec:coreset_matching} (for remote-matching).

\begin{algorithm}[tb]
   \caption{\textsc{PFCoreset}: $O(1)$-approximate remote-pseudoforest composable coreset algorithm}
   \label{alg:pf_coreset}
\begin{algorithmic}[1]
    \STATE {\bfseries Input:} data $X = \{x_1, \dots, x_n\}$, integer $k$, parameter $\eps \in (0, 1]$.
    \IF{$n < 2k^{1+\eps}+k$}
        \STATE \textbf{Return} $X$.
    \ENDIF
    \STATE $Y = \{y_1, \dots, y_k\} \leftarrow \GMM(x_1, \dots, x_n, k)$.
    \FOR{$i=1$ to $n$}
        \STATE $\tilde{r}_i \leftarrow $ $(k+1)^{\text{th}}$ largest value of $\rho(x_i, x_j)$ across all $j \le n$.
    \ENDFOR
    \STATE $\tilde{r} \leftarrow \min_i \tilde{r}_i$, $x \leftarrow \arg\min_i \tilde{r}_i$
    \STATE $U \leftarrow$ $k$ furthest points in $X$ from $x$.
    \STATE $P \leftarrow$ $k$ arbitrary points within distance $\tilde{r}$ of $x$.
    \STATE $S, T \leftarrow \textsc{FindST}(X, k, \eps, \tilde{r})$. \COMMENT{See \Cref{alg:find_ST}}
    \STATE \textbf{Return} $C \leftarrow P \cup S \cup T \cup U \cup Y$.
\end{algorithmic}
\end{algorithm}

\begin{algorithm}[tb]
   \caption{\textsc{MWMCoreset}: $O(1)$-approximate remote-matching composable coreset algorithm}
   \label{alg:mwm_coreset}
\begin{algorithmic}[1]
    \STATE {\bfseries Input:} data $X = \{x_1, \dots, x_n\}$, integer $k$.
    \IF{$n \le 3k$}
        \STATE \textbf{Return} $X$
    \ELSE
        \STATE $Y = \{y_1, \dots, y_k\} \leftarrow \GMM(x_1, \dots, x_n, k)$.
        \STATE \textbf{Initialize} $S_1, \dots, S_k \leftarrow \emptyset$.
        \FOR{$i = 1$ to $n$}
            \STATE Add $i$ to $S_j$ for $j = \arg\min \rho(x_i, y_j)$.
        \ENDFOR
        \STATE \textbf{Initialize} $C \leftarrow Y$.
        \FOR{$i=1$ to $k/2$}
            \STATE \textbf{Find} $x, x' \in X \backslash C$ such that $x, x'$ are in the same $S_j$.
            \STATE $C \leftarrow C \cup \{x, x'\}$.
        \ENDFOR
        \STATE \textbf{Return} $C$.
    \ENDIF
\end{algorithmic}
\end{algorithm}

\section{Offline Remote-Matching Algorithm} \label{sec:offline_matching_body}

In this section, we design $O(1)$-approximate offline algorithms for remote-matching. In this section, we first provide a technical overview, then the algorithm description and pseudocode, and finally we provide the full analysis.

\subsection{Technical Overview} \label{subsec:offline_overview}

The remote-matching offline algorithm first utilizes some simple observations that we made in \Cref{subsec:coreset_overview}. Namely, we may assume the largest minimum-weight matching cost of any subset of $k$ points is $\omega(k \cdot r)$, or else GMM provides an $O(1)$-approximation. Next, if the optimum solution was some $Z = \{z_1, \dots, z_k\},$ we can again consider mapping each $z_i$ to its closest $y_j$, at the cost of having duplicates. However, as noted in \Cref{subsec:coreset_overview}, we may delete a point twice without affecting the matching cost: this means we can keep deleting a point twice until each $y_j$ is only there $0$ times (if the total number of $z_i$'s closest to $y_j$ was even) or $1$ time (if the total number of $z_i$'s closest to $y_j$ was odd).

However, we have no idea what $Z$ actually is, so we have no idea whether each $y_j$ should be included or not. However, this motivates the following simpler problem: among the $k$ points $Y$, choose a (even-sized) subset of $Y$ maximizing the matching cost.

One attempt at solving this problem is to choose $\{y_1, \dots, y_p\}$ for some $p \le k$: this will resemble an argument in~\cite{chandra2001diversitymax}. The idea is that if we define $r_p$ to be the maximum value $\rho(x, \{y_1, \dots, y_p\})$ over all $x \in X$, the same argument as \Cref{prop:gmm} implies that all points among $y_1, \dots, y_p$ are separated by at least $r_p$. Hence, for the best $p$ we can obtain matching cost $\Omega(\max_{1 \le p \le k} p \cdot r_p)$. Conversely, it is known that the minimum-weight-matching cost of any set of points $Z$ is upper-bounded by the cost of the minimum spanning tree of $Z$. But the minimum spanning tree has cost at most $\sum_{p=1}^k r_p$, since we can create a tree by adding an edge from each $y_{p+1}$ to its closest center among $y_1, \dots, y_p$, which has distance $r_p$. Since $\max (p \cdot r_p) \ge \Omega\left(\frac{1}{\log k}\right) \cdot \sum_{p=1}^k r_p$ (with equality for instance if $r_p = \frac{1}{p}$), we can obtain an $O(\log k)$-approximation.

For simplicity, we focus on the case where $r_p = \Theta(1/p)$ for all $1 \le p \le k$. We would hope that either the minimum spanning tree cost of $Y$, which we call $\MST(Y)$, is actually much smaller than $\log k,$ or there is some alternative selection to obtain matching cost $\Omega(\log k)$ rather than $O(1)$. Suppose that $\MST(Y) = \Omega(\log k)$: furthermore, for simplicity suppose the $p$th largest edge of the tree has weight $\frac{1}{p}.$ If we considered the graph on $Y$ connecting two points if their distance is less than $\frac{1}{p}$, it is well-known that the graph must therefore split into $p$ disconnected components.

Now, for some fixed $p$ suppose that we chose a subset $Z$ of $Y$ such that each connected component in the graph above has an odd number of points in $Z$. Then, any matching must send at least one point in each $Z \cap CC_j$ (where $CC_j$ is the $j$th connected component) to a point in a different connected component, forcing an edge of weight at least $\frac{1}{p}$. Since each of $p$ connected components has such a point, together we obtain weight at least $1$. In addition, if we can ensure this property for $p = 2, 4, 8, 16 \dots, k$, we can in fact get there must be at least $2^i$ edges of weight $1/2^i,$ making the total cost $\Omega(\log k)$, as desired.

While such a result may not be possible exactly, it turns out that even a \emph{random} subset of $Y$ satisfies this property asymptotically! Namely, if we choose each point $y \in Y$ to be in $Z$ with $50\%$ probability, each $CC_j$ is odd with $50\%$ probability. So in expectation, for all $p$, the number of connected components of odd size is $p/2$. Even if we make sure $Z$ has even size, this will still be true, replacing $p/2$ with $\Omega(p)$. Since this is true for all $p$ in expectation, by adding over powers of $2$ for $p$, we will find $k$ points with $\Omega(\log k)$ matching cost in expectation.

\subsection{Algorithm Desciption and Pseudocode}

Given a dataset $X = \{x_1, \dots, x_n\},$ we recall that the goal of the remote-matching problem is to find a subset $Z = \{z_1, \dots, z_k\} \subset X$ of $k$ points, such that the minimum-weight matching cost of $Z$, $\MWM(Z)$, is approximately maximized.
In this subsection, we describe our $O(1)$-approximate remote-matching algorithm. We also provide pseudocode in \Cref{alg:mwm_offline}. 

\paragraph{Algorithm Description:} The algorithm proceeds as follows. First, run the GMM algorithm for $k$ steps, to obtain $k$ points $Y = \{y_1, \dots, y_k\} \subset X$. Define the subsets $S_1, \dots, S_k$ as a partitioning of $X,$ where $x \in X$ is in $S_i$ if $y_i$ is the closest point to $x$ in $Y$. (We break ties arbitrarily.) We use the better of the following two options, with the larger minimum-weight matching cost.
\begin{enumerate}
    \item Simply use $Y = \{y_1, \dots, y_k\}$.
    \item Let $\hat{Z} \subset Y$ be a uniformly random subset of $Y$. Initialize $W$ to $\hat{Z}$ if $|\hat{Z}|$ is even, and otherwise initialize $W$ to $\hat{Z} \backslash \hat{z}$ for some arbitrary $\hat{z} \in \hat{Z}$. Now, if there exist two points not in $W \cup Y$ but in the same subset $S_i$, add both of them to $W$. Repeat this procedure until $|W| = k.$
\end{enumerate}
We will use whichever of $Y$ or $W$ has the larger minimum-weight matching cost. Since minimum-weight matching can be computed in polynomial time, we can choose the better of these two in polynomial time.

In Theorem \ref{thm:matching_offline}, we assume $n \ge 3k.$ Because of this assumption, if $|W| < k$, then $|W \cup Y| \le 2k-1$, which means $|X \backslash (W \cup Y)| \ge k+1$. Hence, by Pigeonhole Principle, two of these points must be in the same set $S_i,$ which means that the procedure described above is indeed doable.

\begin{algorithm}[tb]
   \caption{\textsc{MWMOffline}: $O(1)$-approximate remote-matching algorithm}
   \label{alg:mwm_offline}
\begin{algorithmic}[1]
    \STATE {\bfseries Input:} data $X = \{x_1, \dots, x_n\}$, even integer $k$.
    \STATE $Y = \{y_1, \dots, y_k\} \leftarrow \GMM(x_1, \dots, x_n, k)$.
    \STATE \textbf{Initialize} $S_1, \dots, S_k \leftarrow \emptyset$.
    \FOR{$i = 1$ to $n$}
        \STATE Add $i$ to $S_j$ if $j = \arg\min \rho(x_i, y_j)$.
    \ENDFOR
    \STATE $Z \leftarrow$ random subset of $Y$.
    \IF{$|Z|$ is odd}
        \STATE{Remove an arbitrary element from $Z$}
    \ENDIF
    \STATE \textbf{Initialize} $W \leftarrow Z$.
    \WHILE{$|W| < k$}
        \STATE Find some $x, x' \in X \backslash (W \cup Y),$ such that $x, x'$ are in the same subset $S_j$.
        \STATE Add $x, x'$ to $W$.
    \ENDWHILE
    \STATE \textbf{Return} whichever of $Y, W$ has larger minimum-weight matching cost.
\end{algorithmic}
\end{algorithm}

\subsection{Analysis}

The first ingredient in proving Theorem \ref{thm:matching_offline} is the following lemma, which shows that assuming the random subset $Z$ we chose in Line 7 of \Cref{alg:mwm_offline} is sufficiently good, the algorithm produces an $O(1)$-approximation.

\begin{lemma} \label{lem:mwm_lem_2}
    For some constant $\frac{1}{2} \ge \alpha > 0$, suppose that $\MWM(Z) \ge \alpha \cdot \max_{Z' \subset Y: |Z'| \text{ is even}} \MWM(Z')$. Then, Algorithm \ref{alg:mwm_offline} provides a $\frac{4}{\alpha}$-approximation for the remote-matching problem.
\end{lemma}

\begin{proof}
    Let $M$ be the optimal remote-matching cost. Let $r$ be the maximum distance from any point in $X \backslash Y$ to its closest point in $Y$. Note that $\rho(y_i, y_j) \ge r$ for all $i, j \le k,$ by Proposition \ref{prop:gmm}.
    
    First, suppose that $M \le 2 \alpha^{-1} \cdot r \cdot k$. In this case, because every pair in $Y$ has pairwise distance at least $r,$ we have $\MWM(Y) \ge r \cdot \frac{k}{2}.$ Hence, $\MWM(Y) \ge \frac{\alpha}{4} \cdot M,$ which means we have a $\frac{4}{\alpha}$-approximation.

    Alternatively, suppose $M \ge 2 \alpha^{-1} \cdot r \cdot k.$ Let $W_0 \subset X$ be the set of $p$ points that achieves this, i.e., $\MWM(W_0) = M$. Consider the following multiset $\tilde{W}_0$ of size $p$ in $Y$, where each point in $W_0$ is mapped to its closest center in $Y$ (breaking ties in the same way as in the algorithm). Then, every pair of distances between $W_0$ and $\tilde{W}_0$ changes by at most $2r.$ This means \emph{every} matching has its cost change by at most $\frac{k}{2} \cdot 2r = rk$, so $\MWM(\tilde{W}_0) \ge M- r k.$

    Now, let $Z_0 \subset Y$ be the set of points where $y_i \in Z_0$ if and only if $y_i$ is in $\tilde{W}_0$ an odd number of times. Then, $\MWM(\tilde{W}_0) = \MWM(Z_0)$. To see why, first note that $\MWM(\tilde{W}_0) \le \MWM(Z_0)$ since we can convert any matching of $Z_0$ to a matching of $\tilde{W}_0$ by simply matching duplicate points in $\tilde{W}_0$ until only $Z_0$ is left. To see why $\MWM(\tilde{W}_0) \ge \MWM(Z_0)$, note that if an optimal matching of $\tilde{W}_0$ connected some copy of $y$ to a point $y' \neq y$ and another copy of $y$ to a point $y'' \neq y$, we can always replace the edges $(y, y')$ and $(y, y'')$ with $(y, y)$ and $(y', y''),$ which by Triangle inequality will never increase the cost. We may keep doing this until a maximal number of duplicate points are matched together, and only one copy of each element in $Z_0$ will be left. Hence, we have 
\begin{equation} \label{eq:main_matching_1}
    \MWM(Z_0) = \MWM(\tilde{W}_0) \ge M-rk.
\end{equation}

    Similarly, let $Z, W$ be the sets found in the algorithm described above, and let $\tilde{W}$ be the multiset formed by mapping each point in $W$ to its nearest center in $Y$. As in the case with $Z_0$ and $\tilde{W}_0,$ we have that for $Z$ and $\tilde{W}$, a point $z$ is in $Z$ if and only if $z$ is in $\tilde{W}$ an odd number of times. Hence, $\MWM(\tilde{W}) = \MWM(Z)$. Likewise, each point in $\tilde{W}$ has distance at most $r$ from its corresponding point in $W$, which means $\MWM(W) \ge \MWM(\tilde{W}) - rk$. Hence, we have 
\begin{equation} \label{eq:main_matching_2}
    \MWM(Z) = \MWM(\tilde{W}) \le \MWM(W)+rk.
\end{equation}

    Overall, $\MWM(Z) \ge \alpha \cdot \max_{Z \subset Y: |Z| \text{ is even}} \MWM(Z) \ge \alpha \cdot \MWM(Z_0),$ so
\begin{align*}
    \MWM(W) &\ge \MWM(Z)-rk \ge \alpha \cdot \MWM(Z_0) - rk \\
    &\ge \alpha \cdot M - (1+\alpha) rk.
\end{align*}
    But note that $M \ge 2 \alpha^{-1} rk,$ which means that $(1+\alpha) rk \le \frac{\alpha(1+\alpha)}{2} \cdot M \le \frac{3}{4} \cdot \alpha \cdot M$. Hence, $\MWM(W) \ge \frac{\alpha}{4} \cdot M$, which again means we have a $\frac{4}{\alpha}$-approximation.
\end{proof}

The main technical lemma that we will combine with Lemma \ref{lem:mwm_lem_2} shows that $Z$ has the desired property. We now state the lemma, but we defer the proof slightly, to \Cref{subsec:technical_lemma}.
We remark that the proof roughly follows the intuition described at the end of \Cref{subsec:offline_overview}.

\begin{lemma} \label{lem:mwm_main}
    Let $\hat{Z}$ be a random subset of $Y$ where each element is independently selected with probability $1/2$. If $|\hat{Z}|$ is even, we set $Z = \hat{Z}$, and if $|\hat{Z}|$ is odd, we arbitrarily remove $1$ element from $\hat{Z}$ to generate $Z$. Then,
\[\BE[\MWM(Z)] \ge \frac{1}{16} \cdot \max_{Z' \subset Y: |Z'| \text{ is even}} \MWM(Z').\]
\end{lemma}

Given Lemmas \ref{lem:mwm_lem_2} and \ref{lem:mwm_main}, we explain how combine them to prove Theorem~\ref{thm:matching_offline}.

\begin{proof}[Proof of Theorem \ref{thm:matching_offline}]
    Suppose we generate a random subset $Z$ of $Y$ (possibly removing an element), and suppose that $\MWM(Z) = \alpha \cdot \max_{Z \subset Y: |Z| \text{ is even}} \MWM(Z)$. Then, the output of the algorithm has matching cost at least $\frac{\alpha}{4}$ times the optimum $k$-matching cost $\MWM_k(X)$, by Lemma \ref{lem:mwm_lem_2}. However, by Lemma \ref{lem:mwm_main}, $\BE[\alpha] \ge \frac{1}{16}$, which means that the expected matching cost of the output is $\frac{\BE[\alpha]}{4} \cdot \MWM_k(X) \ge \frac{1}{64} \cdot \MWM_k(X)$.

    If we want this to occur with high probability, note that the matching cost of the output can never be more than $\MWM_k(X)$. Hence, by Markov's inequality, with at least $\frac{1}{64^2} = \frac{1}{4096}$ probability, the output has matching cost at least $\frac{1}{65} \cdot \MWM_k(X)$. If we repeat this $O(1)$ times and return the set with best matching cost, we can find a set of size $k$ with matching cost at least $\frac{1}{65} \cdot \MWM_k(X),$ with probability at least $0.99$.
\end{proof}

Before proving \Cref{lem:mwm_main}, we will need some additional preliminaries.

\subsection{Preliminaries for Lemma \ref{lem:mwm_main}} \label{subsec:prelim_technical_lemma}

To prove Lemma \ref{lem:mwm_main}, we will need several preliminary facts relating to the cost of a minimum weight matching, as well as the cost of a minimum spanning tree of a set of points.

First, we have the following fact, bounding the minimum weight matching in terms of the MST.

\begin{proposition} \label{prop:mwm_mst} (Classical, see Proof of Lemma 5.2 in \cite{chandra2001diversitymax})
    For any (finite, even sized) set of data points $Z$ in a metric space, $\MWM(Z) \le \MST(Z)$.
\end{proposition}

Next, given a subset $Z$ of $Y$ in a metric space, we can bound the minimum spanning tree cost of $Z$ in terms of the minimum spanning cost of $Y$.

\begin{proposition} \label{prop:steiner} (Classical, see \cite{gilbert1968steiner})
    Let $Z \subset Y$ be (finite) sets of data points in some metric space. Then, $\MST(Z) \le 2 \cdot \MST(Y)$.
\end{proposition}

Next, we equate the minimum spanning tree of a dataset $Y$ with the number of connected components in a family of graphs on $Y$. The following proposition essentially follows from the same argument as in \cite[Lemma 2.1]{czumaj2009sublinearmst}. We prove it here for completeness since the statement we desire is not explicitly proven in \cite{czumaj2009sublinearmst}.

\begin{proposition} \label{prop:mst_cc_count}
    Given a dataset $Y$ in a metric space and a radius $r > 0$, define $G_r(Y)$ to be the graph on $Y$ that connects two data points if and only if their distance is at most $r$. Define $P_r(Y)$ to be the number of connected components in $G_r(Y).$ Then,
\[\MST(Y) \in \left[\frac{1}{2}, 1\right] \cdot \left(\sum_{i \in \BZ} 2^i \cdot (P_{2^i}(Y)-1)\right).\]
\end{proposition}

\begin{proof}
    Note that if $2^i$ is at least $\diam(Y),$ the diameter of $Y$, $P_{2^i}(Y) = 1,$ which means we may ignore the summation for $i$ with $2^i > \diam(Y)$. Hence, by scaling by some power of $2$, we may assume WLOG that $\diam(Y) < 1,$ and that the summation is only over $i < 0$.

    Now, for any $t \ge 0,$ let $Q_t(Y)$ be the number of edges in the MST of $Y$ with weight at most $2^{-t}$ and strictly more than $2^{-(t+1)}$ (assuming we run Kruskal's algorithm for MST). Note that $R_t(Y) := \sum_{t' \ge t} Q_{t'}(Y)$ is the number of edges with weight at most $2^{-t}$. Note that $R_t(Y)$ is precisely $n-P_{2^{-t}}(Y)$. To see why, note that the $R_t(Y)$ edges form a forest with $n-R_t(Y)$ connected components. In addition, in the graph $G_{2^{-t}}(Y)$, none of the $n-R_t(Y)$ components can be connected to each other, or else there would have been another edge of weight at most $2^{-t}$ that Kruskal's algorithm would have had to add. Therefore, $\sum_{t' \ge t} Q_{t'}(Y) = n-P_{2^{-t}}(Y)$. By subtracting this equation from the same equation replacing $t$ with $t+1$, we obtain
\begin{equation} \label{eq:czumaj_sohler_1}
    Q_t(Y) = P_{2^{-(t+1)}}(Y)-P_{2^{-t}}(Y).
\end{equation}

    Now, note that by definition of $Q_t(Y),$ the cost of $\MST(Y)$ is between $\sum_{t \ge 0} 2^{-(t+1)} \cdot Q_t(Y)$ and $\sum_{t \ge 0} 2^{-t} \cdot Q_t(Y)$. Equivalently, it equals $\alpha \cdot \left(\sum_{t \ge 0} 2^{-t} \cdot Q_t(Y)\right),$ for some $\alpha \in [1/2, 1]$. Therefore,
\begin{align*}
    \MST(Y) &= \alpha \cdot \left(\sum_{t \ge 0} 2^{-t} \cdot Q_t(Y)\right) \\
    &= \alpha \cdot \left(\sum_{t \ge 0} 2^{-t} \cdot \left(P_{2^{-(t+1)}}(Y)-P_{2^{-t}}(Y)\right)\right) \\
    &= \alpha \cdot \left(\sum_{t \ge 0} (2^{-t}-2^{-(t+1)}) P_{2^{-(t+1)}}(Y) - P_1(Y)\right) \\
    &= \alpha \cdot \left(\sum_{t \ge 1} 2^{-t} P_{2^{-t}}(Y) - 1\right) \\
    &= \alpha \cdot \left(\sum_{t \ge 1} 2^{-t} (P_{2^{-t}}(Y) - 1)\right).
\end{align*}
    The second-to-last line follows since the diameter is at most $1$ so $G_1(Y)$ has one connected component, and the last line follows because $\sum_{t \ge 1} 2^{-t} = 1.$
\end{proof}

Finally, we need to consider the minimum weight matching cost in a \emph{hierarchically well-separated tree} (HST). 

\begin{definition} \label{def:hst}
A \emph{hierarchically-well seprated tree (HST)} is a depth-$d$ tree (for some integer $d \ge 1$) with the root as depth $0$, and every leaf has depth $d$. For any node $u$ in the tree of depth $1 \le t \le d,$ each edge from $u$ to its parent has weight $2^{-t}$. For two nodes $v, w$ in the HST, the distance $d_{\HST}(v, w)$ is simply the sum of the edge weights along the shortest path from $v$ to $w$ in the tree.
\end{definition}

Note that for any two leaf nodes $v, w$ in an HST, if their least common ancestor has depth $t$, the distance between $v$ and $w$ is $2 \cdot \left(2^{-(t+1)} + 2^{-(t+2)} + \cdots + 2^{-d}\right) = 2 \cdot (2^{-t}-2^{-d})$.

We will make use of the following result about points in an HST metric.

\begin{proposition} \label{prop:emd_odd_count} \cite[Claim 3]{indyk2004emd}
    Let $Z$ be a (finite, even sized) set of points that are leaves in a depth-$d$ HST. Let $m_i$ be the number of nodes at level $i$ with an \emph{odd} number of descendants in $Z$. Then, with respect to the HST metric, the minimum weight matching cost equals
\[\MWM = \sum_{i=0}^d 2^{-i} \cdot m_i.\]
\end{proposition}

We remark that the corresponding statement in \cite{indyk2004emd} has an additional additive factor of $n = |Z|$ in the right-hand side. This is because we include the bottom level in our sum (which consists of $n$ nodes each with exactly one descendant), whereas \cite{indyk2004emd} does not.

\subsection{Proof of Lemma \ref{lem:mwm_main}} \label{subsec:technical_lemma}

We are now ready to prove \Cref{lem:mwm_main}

\begin{proof}[Proof of \Cref{lem:mwm_main}]
    Assume WLOG (by scaling) that the diameter of $Y$ is at most $1$. Let $Z$ be a subset of $Y$ with even size. By Proposition \ref{prop:mwm_mst}, $\MWM(Z) \le \MST(Z)$. By Proposition \ref{prop:steiner}, $\MST(Z) \le 2 \cdot \MST(Y)$. Combining these together, we have
\begin{equation} \label{eq:mwm_mst}
    \max_{Z \subset Y: |Z| \text{ even}} \MWM(Z) \le 2 \cdot \MST(Y).
\end{equation}

    Now, for our dataset $Y$ and any positive real $r > 0,$ recall that $G_r(Y)$ is defined as the graph on $Y$ that connects two data points if their distance is at most $r$. In addition, define $\mathcal{P}_r(Y)$ to be the partitioning of $Y$ into connected components based on $G_r(Y)$, and recall that $P_r(Y) = |\mathcal{P}_r(Y)|$ equals the number of connected components in $G_r(Y)$. Then, Proposition \ref{prop:mst_cc_count} tells us that
\begin{equation} \label{eq:mst_ub}
    \MST(Y) \le \sum_{i \in \BZ} 2^i \cdot (P_{2^i}(Y)-1).
\end{equation}

    Now, consider the following ``embedding'' of $Y$ into a depth-$d$ hierarchically well-separated tree (where we will choose $d$ later) as follows. By scaling, assume WLOG that the diameter of $Y$ is $1$. For each integer $0 \le t \le d$, the nodes at level $t$ will be the connected components in $G_{2^{-t}}(Y)$, where the children of any node at depth $t$, represented by a subset $Z$ of $Y$, are simply the connected components in $\mathcal{P}_{2^{-(t+1)}}(Y)$ contained in $Z$.

    The distance $d_{\HST}(y_i, y_j)$ between any two vertices $y_i, y_j$ in the HST is precisely $2(2^{-t}-2^{-d})$ if $y_i, y_j$ have common ancestor at level $t$ of the HST. Note that if $d_{\HST}(y_i, y_j) = 2(2^{-t}-2^{-d}),$ then $y_i, y_j$ are not in the same connected component of $G_{2^{-(t+1)}},$ which means that $\rho(y_i, y_j) > 2^{-(t+1)}$. Importantly, this means that $d_{\HST}(y_i, y_j) \le 4 \rho(y_i, y_j)$ for all pairs $i, j$. Hence, for any subset $Z \subset Y$ of even size, the minimum weight matching cost $\MWM_{\HST}(Z)$ with respect to the HST metric is at most $4$ times the true minimum weight matching cost, i.e.,
\begin{equation} \label{eq:hst_mwm}
    \MWM_{\HST}(Z) \le 4 \cdot \MWM(Z).
\end{equation}

    Finally, we consider selecting a random subset $\hat{Z} \subset Y$, and provide a lower bound for $\MWM_{\HST}(Z),$ where $Z = \hat{Z}$ if $|\hat{Z}|$ is even and otherwise $Z$ equals $\hat{Z}$ after removing a single (arbitrary) element. Note that for each node $v$ of depth $t$, corresponding to a connected component in $\mathcal{P}_{2^{-t}}(Y),$ the probability that it has an odd number of descendants in $\hat{Z}$ if $\hat{Z}$ is picked at random is precisely $1/2$. This implies that the expectation of $\sum_{i = 0}^{d} 2^{-i} \cdot m_i,$ where $m_i$ is the number of nodes at level $i$ with an odd number of descendants in $\hat{Z}$, is $\frac{1}{2} \cdot \sum_{i=0}^d 2^{-i} \cdot P_{2^{-i}}(Y)$.

    Note, however, that $\hat{Z}$ has odd size with $1/2$ probability. In this event, we remove an arbitrary element of $\hat{Z}$, which may reduce each $m_i$ by $1$. This only happens with $50\%$ probability, so after this potential removal of a point, the expectation of $\sum_{i = 0}^{d} 2^{-i} \cdot m_i(Z),$ where $m_i(Z)$ is the number of nodes at level $i$ with an odd number of descendants in $Z$, is at least $\frac{1}{2} \cdot \sum_{i=0}^d 2^{-i} \cdot (P_{2^{-i}}(Y)-1)$. Since $\diam(Y)$ is at most $1$, this implies $P_{2^i}(Y)-1 = 0$ for all $i \ge 1$. Also, for $i > d,$ $2^{-i} \cdot (P_{2^{-i}}(Y)-1) \le 2^{-i} \cdot n$. If we sum this up over all $i > d$, this is still at most $2^{-d} \cdot n$. Hence, we have that
\begin{equation} \label{eq:random_large_mwm_cost}
    \BE_Z\left[\sum_{i=0}^d 2^{-i} \cdot m_i(Z)\right] \ge \frac{1}{2} \cdot \left(\sum_{i \in \BZ} 2^i (P_{2^i}(Y)-1)\right) - n \cdot 2^{-d}.
\end{equation}

    In summary, we have that
\begin{alignat*}{3}
    \max_{Z' \subset Y: |Z'| \text{ even}} \MWM(Z') 
    & \le 2 \cdot \MST(Y) && \hspace{0.5cm} \text{By Equation \eqref{eq:mwm_mst}} \\
    & \le 2 \cdot \sum_{i \in \BZ} 2^i \cdot \left(P_{2^i}(Y)-1\right) && \hspace{0.5cm} \text{By Equation \eqref{eq:mst_ub}} \\
    & \le 4 \cdot \left(\BE_Z\left[\sum_{i=0}^d 2^{-i} \cdot m_i(Z)\right] + n \cdot 2^{-d}\right) && \hspace{0.5cm} \text{By Equation \eqref{eq:random_large_mwm_cost}} \\
    & = 4 \cdot \left(\BE_Z\left[\MWM_{\HST}(Z)\right] + n \cdot 2^{-d}\right) && \hspace{0.5cm} \text{By Proposition \ref{prop:emd_odd_count}} \\
    & \le 16 \cdot \left(\BE_Z\left[\MWM(Z)\right] + n \cdot 2^{-d}\right). && \hspace{0.5cm} \text{By Equation \eqref{eq:hst_mwm}}
\end{alignat*}
    We can choose the depth of the HST to be arbitrarily large, which therefore implies that 
\[\BE_{Z}[\MWM(Z)] \ge \frac{1}{16} \cdot \max_{Z' \subset Y: |Z'| \text{ even}} \MWM(Z'). \qedhere \]
\end{proof}

\section{Composable Coreset for Remote-Pseudoforest} \label{sec:coreset_pseudoforest}

In this section, we describe and analyze 
the composable coreset algorithm for remote-pseudoforest.

\subsection{Algorithm}

In this subsection, we prove why the algorithm given in \Cref{alg:pf_coreset} creates an $O(1)$-approximate composable coreset.
First, we describe the algorithm in words. We recall that we have $m$ datasets $X^{(1)}, \dots, X^{(m)}$: we wish to create a coreset $C^{(j)}$ of each $X^{(j)}$ so that $\bigcup_{j=1}^m C^{(j)}$ contains a set $Z$ of $k$ points such that $\PF(Z) \ge \Omega(1) \cdot \PF_k\left(\bigcup_{j=1}^m X^{(j)}\right)$.

\paragraph{Coreset Construction:}
Suppose that $|X^{(j)}| \ge 2 k^{1+\eps} + k$. Let $\tilde{r}^{(j)}$ represent the smallest value such that there exists a ball $\tilde{B}^{(j)}$ of radius $\tilde{r}^{(j)}$ around some $x^{(j)} \in X^{(j)}$ that contains all but at most $k$ of the points in $X^{(j)}$.
The point $\tilde{r}^{(j)}$ and a corresponding $x^{(j)}, \tilde{B}^{(j)}$ can be found in $O(|X^{(j)}|^2)$ time.

Choose $U^{(j)}$ to be the set of $k$ points furthest from $x^{(j)}$. These will either be precisely the $k$ points outside $\tilde{B}^{(j)}$, or all of the points outside $\tilde{B}^{(j)}$ plus some points on the boundary of $\tilde{B}^{(j)}$, to make a total of $k$ points. In addition, we choose some arbitrary set $P^{(j)} \subset X^{(j)}$ of any $k$ points in the ball $B^{(j)}$.

Next, we will choose sets $S^{(j)}, T^{(j)} \subset X^{(j)}$ of size $k$, such that $\rho(S^{(j)}, T^{(j)}) \ge \frac{\eps}{2} \cdot \tilde{r}^{(j)}$. It is not even clear that such sets exist, but we will show how to algorithmically find such sets in $O(|X^{(j)}|^2)$ time (assuming $|X^{(j)}| \ge 2 k^{1+\eps}+k$).

Finally, we run the GMM algorithm on $X^{(j)}$ to obtain $Y^{(j)} = \{y_1^{(j)}, y_2^{(j)}, \dots, y_k^{(j)}\}$. The final coreset will be $C^{(j)} := P^{(j)} \cup S^{(j)} \cup T^{(j)} \cup U^{(j)} \cup Y^{(j)}$. Note that each of $P^{(j)}, S^{(j)}, T^{(j)}, U^{(j)}, Y^{(j)}$ has size at most $k$, so $|C^{(j)}| \le 5k$.

Alternatively, if $|X^{(j)}| < 2k^{1+\eps}+k,$ we choose the coreset to simply be $X^{(j)}$. For convenience, in this setting, we define $U^{(j)} := X^{(j)}$ and $P^{(j)}, S^{(j)}, T^{(j)}, Y^{(j)}$ to all be empty.

\medskip

We have not yet described how to find $S^{(j)}, T^{(j)}$, let alone prove they even exist. We now describe an $O(|X^{(j)}|^2)$ time algorithm that finds $S^{(j)}, T^{(j)} \subset X^{(j)}$ of size $k$, such that $\rho(S^{(j)}, T^{(j)}) \ge \frac{\eps}{2} \cdot \tilde{r}^{(j)}$.

\paragraph{Efficiently finding $S^{(j)}, T^{(j)}$:}
Here, we show that we can efficiently find $S^{(j)}, T^{(j)}$ from the coreset construction, as long as $|X^{(j)}| \ge 2 k^{1+\eps}+k$. We formalize this with the following lemma.

\begin{lemma} \label{lemma:find_ST}
    Let $\eps > 0$ be a fixed constant, and consider a dataset $X$ of size at least $2 k^{1+\eps} + k$, and suppose that no ball of radius smaller than $\tilde{r}$ around any $x \in X$ contains all but at most $k$ points in $X$. (In other words, for every $x \in X$, there are at least $k$ points in $X$ of distance at least $\tilde{r}$ from $x$.) Then, we in $O(|X|^2)$ time, we can find two disjoint sets $S, T \subset X$, each of size $k$, such that $\rho(S, T) \ge \frac{\eps \cdot \tilde{r}}{2}$.
\end{lemma}

\begin{proof}
The algorithm works as follows. First, define $r' = \frac{\eps \cdot \tilde{r}}{2}$. For each point $x \in X$, and for every nonnegative integer $i$, define $N_i(x)$ as the set of points in $X$ of distance at most $i \cdot r'$ from $x$.
We can compute the set $N_i(x)$ for all $x \in X$ and $0 \le i \le 2/\eps$ in $O(|X|^2)$ time, as $\eps$ is a constant.
Suppose there exists $x \in X$ such that $|N_{1/\eps}(x)| \ge k$. By our assumption on $\tilde{r}$, there are at least $k$ points of distance at least $\tilde{r} = \frac{2}{\eps} \cdot r'$ from $x$ (or else we could have chosen $\tilde{r}$ to be smaller). Therefore, we can let $S$ be a subset of size $k$ from $N_{1/\eps}(x)$ and $T$ be a subset of size $k$ of points of distance at least $\frac{2}{\eps} \cdot r'$ from $x$. The minimum distance between any $s \in S$ and $t \in T$ is at least $r' \cdot \left(\frac{2}{\eps} - \frac{1}{\eps}\right) = \frac{\tilde{r}}{2}$.

Alternatively, every $x \in X$ satisfies $|N_{1/\eps}(x)| < k$. Now, consider the following peeling procedure. Let $X_0 := X$: for each $h \ge 1$, we will inductively create $X_h \subsetneq X_{h-1}$ from $X_{h-1}$, as follows. 
First, we pick an arbitrary point $x_h \in X_{h-1}$. For any point $x \in X_{h-1}$ and any integer $i \ge 0$, define $N_i(x_h; X_{h-1}) = N_i(x_h) \cap X_{h-1}$ to be the set of points of distance at most $i \cdot r'$ from $x_h$ in $X_{h-1}$. By our assumption, we have that $|N_{1/\eps}(x_h; X_{h-1})| \le |N_{1/\eps}(x_h)| < k$, so there exists some $i(h)$ with $0 \le i(h) \le \frac{1}{\eps}$, such that $\frac{|N_{i(h)+1}(x_h, X_{h-1})|}{|N_{i(h)}(x_h, X_{h-1})|} \le k^{\eps}$ and $|N_{i(h)}(x_h, X_{h-1})| \le k$. Choose such an $i = i(h)$, and let $X_h := X_{h-1} \backslash N_{i(h)}(x_{h}, X_{h-1})$.

We repeat this process until we have found the first $X_\ell$ with $|X \backslash X_\ell| \ge k$. Note that each removal process removes at least $1$ and at most $k$ elements, so $|X \backslash X_\ell| \le 2 k$. Let $S_h = N_{i(h)}(x_h, X_{h-1}) = X_{h-1} \backslash X_h$ for each $1 \le h  \le \ell$, so $X \backslash X_\ell = S_1 \cup \cdots \cup S_\ell$.
Note, however, that any point within distance $r'$ of some $x \in S_h$ was either in $S_1 \cup \cdots \cup S_{h-1}$, or was in $N_{i(h)+1}(x_h, X_{h-1})$. 
In other words, every point of distance $r'$ of $x \in S_1 \cup \cdots \cup S_\ell$ is in $\bigcup_{h=1}^{\ell} N_{i(h)+1}(x_h, X_{h-1})$. But this has size at most 
\[\sum_{i=1}^{\ell} k^{\eps} \cdot |N_{i(h)}(x_h, X_{h-1})| = k^\eps \cdot \sum_{i=1}^{\ell} |S_i| \le k^{\eps} \cdot 2k = 2k^{1+\eps}.\]
So, assuming that $|X| \ge 2k^{1+\eps}+k$, defining $S := S_1 \cup \cdots \cup S_\ell,$ we have that $|S| \ge k$ and there are at least $k$ points in $X$ that are \emph{not} within distance $r' = \frac{\eps \cdot \tilde{r}}{2}$ of $S$.
\end{proof}

We include pseudocode for the algorithm described in the proof of \Cref{lemma:find_ST}, in \Cref{alg:find_ST}.

\begin{algorithm}[tb]
   \caption{\textsc{FindST}: Find two sets $S, T$ of size $k$ with large $\rho(S, T)$}
   \label{alg:find_ST}
\begin{algorithmic}[1]
    \STATE {\bfseries Input:} data $X = \{x_1, \dots, x_n\}$, integer $k$, parameter $\eps \in (0, 1]$, radius $\tilde{r}$.
    \STATE $r' \leftarrow \frac{\eps \cdot \tilde{r}}{2}$.
    \FOR{$x$ in $X$}
        \FOR{$i = 0$ to $2/\eps$}
            \STATE $N_i(x) = \{z \in X: \rho(x, z) \le i \cdot r'\}$.
        \ENDFOR
        \IF{$|N_{1/\eps}(x)| \ge k$}
            \STATE $S \leftarrow$ arbitrary subset of size $k$ in $N_{1/\eps}(x)$.
            \STATE $T \leftarrow$ arbitrary $k$ points of distance at least $\tilde{r}$ from $x$.
            \STATE \textbf{Return} $S, T$.
        \ENDIF
    \ENDFOR
    \STATE $S = \emptyset, X_0 \leftarrow X, h \leftarrow 0$
    \WHILE{$|S| < k$}
        \STATE $h \leftarrow h+1$
        \STATE $x_h \in X_{h-1}$ chosen arbitrarily.
        \STATE Find $0 \le i \le \frac{1}{\eps}$ such that $\frac{|N_{i+1}(x_h) \cap X_{h-1}|}{|N_i(x_h) \cap X_{h-1}|} \le k^{\eps}$.
        \STATE $S_h \leftarrow N_i(x_h) \cap X_{h-1}$.
        \STATE $X_h \leftarrow X_{h-1} \backslash S_h$, $S \leftarrow S \cup S_h$
    \ENDWHILE
    \STATE $S \leftarrow $ arbitrary subset of size $k$ in $S$
    \STATE $T \leftarrow $ arbitrary subset of $k$ points of distance at least $r'$ from all points in $S$.
    \STATE \textbf{Return} $S, T$.
\end{algorithmic}
\end{algorithm}

\subsection{Analysis}

In this section, we prove that the algorithm indeed generates an $O(1/\eps)$-approximate composable coreset of size at most $O(k^{1+\eps})$. By making $\eps$ an arbitrarily small constant, this implies we can find a constant-approximate composable coreset of size $O(k^{1+\eps})$ for any arbitrarily small constant $\eps$.

Let $\OPT$ represent the optimal set of $k$ points in all of $X = \bigcup_{j=1}^m X^{(j)}$, that maximizes remote-pseudoforest cost.
Our goal is to show that there exists a set of $k$ points in $\bigcup_j (P^{(j)} \cup S^{(j)} \cup T^{(j)} \cup U^{(j)} \cup Y^{(j)})$ with pseudoforest cost at least $\Omega(\PF(\OPT))$.

Let $r^{(j)}$ represent the maximum distance from any point in $X^{(j)}$ to its closest point in $Y^{(j)}$.
Note that by \Cref{prop:gmm}, $\PF(Y^{(j)}) \ge k \cdot r^{(j)}$. Hence, if $\PF(\OPT) < 10 \cdot k \cdot \max_j r^{(j)}$, there exists some choice of $j$ with $\PF(Y^{(j)}) \ge 0.1 \cdot \OPT$ and $|Y^{(j)}| = k$. Hence, we get a constant-factor approximation in this case. Otherwise, we may assume that $\PF(\OPT) \ge 10 \cdot k \cdot \max_j r^{(j)}$.

Next, for a fixed $X^{(j)}$, let $X_i^{(j)}$ represent the set of points in $X^{(j)}$ closest to $y_i^{(j)}$ among all points in $Y^{(j)}$.
Given the optimal solution $\OPT$ of $k$ points, let $\OPT_i^{(j)} = \OPT \cap X_i^{(j)}$.

Now, we will define sets $G_i^{(j)}, {G'_i}^{(j)}$ based on the following cases.
\begin{enumerate}
    \item If $|X^{(j)}| < 2k^{1+\eps}+k$, define $G_i^{(j)}={G_i'}^{(j)} = \OPT_i^{(j)}$ for all $i \le k$.
    \item Else, if $y_i^{(j)} \in \OPT_i^{(j)}$, define $G_i^{(j)}$ as $y_i^{(j)} \cup (U^{(j)} \cap \OPT_i^{(j)})$ and ${G'_i}^{(j)} = \OPT_i^{(j)}$.
    \item Else, if $\OPT_i^{(j)} \backslash U^{(j)} = \emptyset$ (i.e., all points in $\OPT_i^{(j)}$ happen to be in $U^{(j)}$), define $G_i^{(j)} = {G'_i}^{(j)} = \OPT_i^{(j)}.$
    \item Else, define $G_i^{(j)} = y_i^{(j)} \cup (U^{(j)} \cap \OPT_i^{(j)})$, and define ${G'_i}^{(j)}$ as $\OPT_i^{(j)}$ with the slight modification of moving a single (arbitrary) point in $\OPT_i^{(j)} \backslash U^{(j)}$ to $y_i^{(j)}$.
\end{enumerate}
We will define the sets $G = \bigcup_{i, j} G_i^{(j)}$ and $G' = \bigcup_{i, j} {G_i'}^{(j)}$.

Importantly, the following five properties always hold for all $i \le m, j \le k$. (They even hold in the setting when $|X^{(j)}| < 2k^{1+\eps}+k$, because we defined $U^{(j)} = X^{(j)}$ and $G_i^{(j)} = {G'_i}^{(j)} = \OPT_i^{(j)}$.)
\begin{enumerate}
    \item $|{G'_i}^{(j)}| = |\OPT_i^{(j)}|$. This means that $|G'| = k$. \label{prop:1}
    \item $G_i^{(j)} \subset {G'_i}^{(j)} \subset X_i^{(j)}$. This means that $G \subset G'$. \label{prop:2}
    \item $G_i^{(j)} \subset U^{(j)} \cup Y^{(j)}$. This means that $G \subset \bigcup_j (U^{(j)} \cup Y^{(j)})$. \label{prop:3}
    \item Every point in ${G'_i}^{(j)} \backslash G_i^{(j)}$ is not in $U^{(j)}$. This means that $\bigcup U^{(j)}$ and $G' \backslash G$ are disjoint. \label{prop:4}
    \item If ${G'_i}^{(j)} \backslash G_i^{(j)}$ is nonempty, then $y_i^{(j)} \in G_i^{(j)}$, so $G_i^{(j)}$ is also nonempty. \label{prop:5}
\end{enumerate}

Now, note that from changing $\OPT$ to $G'$, we never move a point by more than $\max_j r^{(j)}$, which means that $|\PF(G') - \PF(\OPT)| \le 2 k \cdot \max_j r^{(j)}.$ As we are assuming that $\PF(\OPT) \ge 10 k \cdot \max_j r^{(j)}$, we have $\PF(G') \ge 0.8 \cdot \OPT$. Next, if a point $x$ is in ${G'_i}^{(j)}$ but not in $G_i^{(j)}$, then $x \in X^{(j)} \backslash U^{(j)}$ and $y_i^{(j)} \in G_i^{(j)}$, which means that the cost of $x$ with respect to $G'$ is at most $\max_j r^{(j)}$. So for any set $A$, if we define $\cost_A(x)$ for $x \in A$ to denote $\min_{y \in A: y \neq x} \rho(x, y)$, then 
\begin{equation} \label{eq:G'_cost}
    \sum_{x \in G} \cost_{G'}(x) \ge \PF(G') - \sum_{x \in G' \backslash G} \cost_{G'}(x) \ge \PF(G') - k \cdot \max_j r^{(j)} \ge 0.7 \cdot \OPT.
\end{equation}

We now try to find a set $G'' \supset G$ of size $k$ with large pseudoforest cost, but this time we must ensure that $G'' \subset \bigcup_j (P^{(j)} \cup S^{(j)} \cup T^{(j)} \cup U^{(j)} \cup Y^{(j)})$. In other words, to finish the analysis, it suffices to prove the following lemma.

\begin{lemma} \label{lem:better-coreset-main-analysis}
    There exists $G'' \subset \bigcup_j (P^{(j)} \cup S^{(j)} \cup T^{(j)} \cup U^{(j)} \cup Y^{(j)})$ of size at least $k$, such that $G'' \supset G$ and $\sum_{x \in G} \cost_{G''}(x) \ge \Omega(\eps) \cdot \PF(\OPT)$.
\end{lemma}

To see why \Cref{lem:better-coreset-main-analysis} is sufficient to prove \Cref{thm:pseudoforest_coreset}, since $|G| \le k$ and $|G''| \ge k$ we can choose a set $\hat{G}$ of size $k$ such that $G \subseteq \hat{G} \subseteq G''$. Then, $\hat{G} \subset \bigcup_j (P^{(j)} \cup S^{(j)} \cup T^{(j)} \cup U^{(j)} \cup Y^{(j)})$ and
\[\PF(\hat{G}) = \sum_{x \in \hat{G}} \cost_{\hat{G}}(x) \ge \sum_{x \in G} \cost_{\hat{G}}(x) \ge \sum_{x \in G} \cost_{G''}(x),\]
where the last inequality holds because the cost of $x$ never increases from $\hat{G}$ to a larger set $G''$. Finally, by \Cref{lem:better-coreset-main-analysis}, this means $\PF(\hat{G}) \ge \Omega(\eps) \cdot \PF(\OPT)$, as desired.

We now prove \Cref{lem:better-coreset-main-analysis}. First, we show how to construct $G''$.
Let $g = |G|$: note that $g \le k$. If $g = k$, then in fact $G = G'$ and we can set $G'' = G$, which completes the proof by \eqref{eq:G'_cost} and Property \ref{prop:3}. 

Hence, from now on, we may assume that $g < k$.
Recall that $X^{(j)} \backslash U^{(j)}$ is contained in a ball of radius $\tilde{r}^{(j)}$. Next, let $\tilde{r}$ be the radius of $\bigcup_j (X^{(j)} \backslash U^{(j)})$. (Note that for $|X^{(j)}| < 2k^{1+\eps}+k$, $X^{(j)} \backslash U^{(j)}$ is empty.)
We claim the following proposition.

\begin{proposition} \label{prop:jj'}
    There exist $j, j' \le m$, possibly equal, such that $\rho(S^{(j)}, T^{(j')}) \ge \frac{\eps}{10} \cdot \tilde{r}$.
\end{proposition}

\begin{proof}
    Let $A \subset [m]$ be the subset of indices $j$ such that $|X^{(j)}| \ge 2k^{1+\eps}+k$.
    Suppose that $\tilde{r} \le 5 \cdot \max_{j \in A} \tilde{r}^{(j)}$. Then, by setting $j = j'$ to be $\arg\max_{j \subset A} \tilde{r}^{(j)},$ we have that $\rho(S^{(j)}, T^{(j')}) \ge \frac{\eps}{2} \cdot \max_{j \in A} \tilde{r}^{(j)} \ge \frac{\eps}{10} \cdot \tilde{r}$.

    Otherwise, $\tilde{r} > 5 \cdot \max_{j \in A} \tilde{r}^{(j)}$. So, if we pick $j$ arbitrarily, the distance between the center of the ball $\tilde{B}^j$ and the furthest center $\tilde{B}^{j'}$ must be at least $0.8 \cdot \tilde{r}$, or else the ball of radius $0.8 \cdot \tilde{r} + \max_{j \in A} \tilde{r}^{(j)} < \tilde{r}$ around the center of $\tilde{B}^j$ contains all of $\bigcup_j (X^{(j)} \backslash U^{(j)})$. Then, $d(S^{(j)}, T^{(j')}) \ge 0.8 \cdot \tilde{r} - \tilde{r}^{(j)} - \tilde{r}^{(j')} \ge 0.4 \cdot \tilde{r}$, which is at least $\frac{\eps}{10} \cdot \tilde{r}$.
\end{proof}

We now prove \Cref{lem:better-coreset-main-analysis}.
\begin{proof}
    Recall that we already proved the lemma in the case that $G = G'$. So, we may assume $|G| < k$ and $G' \backslash G$ is nonempty.
    We claim that we can set $G''$ to be one of $G \cup P^{(j)}$, $G \cup S^{(j)}$, or $G \cup T^{(j')}$, for $j, j'$ chosen in \Cref{prop:jj'}. 
    
    First, note that $P^{(j)}, S^{(j)}$, and $T^{(j')}$ have size $k$, so all three choices of $G''$ have size at least $k$.

    First, note that $G' \backslash G$ is assumed to be nonempty, which means it is contained in $\bigcup_j (X^{(j)} \backslash U^{(j)})$ by Property \ref{prop:4}, which is contained in the radius $\tilde{r}$ ball. Therefore, $G' \backslash G$ has nonempty intersection with $X \backslash (\bigcup_j U^{(j)})$. 
    Now, fix any point $x \in G$. If $\cost_{G'}(x) \ge 3 \cdot \tilde{r}$, then because $G' \backslash G$ has a point in the radius $\tilde{r}$ ball containing $X \backslash (\bigcup_j U^{(j)})$ (and this point is not $x$ since $x \in G$), $x$ has distance at least $\tilde{r}$ from the radius $\tilde{r}$ ball. So, $\cost_{G \cup P^{(j)}}(x) \ge \cost_{G'}(x) - 2 \tilde{r} \ge \frac{1}{3} \cdot \cost_{G'}(x)$.
    Alternatively, if $\cost_{G'}(x) < 3 \cdot \tilde{r}$, then $\cost_{G \cup S^{(j)}}(x) \ge \min(\cost_{G'}(x), \rho(x, S^{(j)}))$ and likewise, $\cost_{G \cup T^{(j')}}(x) \ge \min(\cost_{G'}(x), \rho(x, T^{(j')}))$. So, 
\[\cost_{G \cup S^{(j)}}(x)+\cost_{G \cup T^{(j')}}(x) \ge \min(\cost_{G'}(x), \rho(S^{(j)}, T^{(j')})) \ge \min\left(\cost_{G'}(x), \frac{\eps}{10} \cdot \tilde{r}\right) \ge \frac{\eps}{30} \cdot \cost_{G'}(x).\]
    In all cases, we have that
\[\cost_{G \cup P^{(j)}}(x) + \cost_{G \cup S^{(j)}}(x) + \cost_{G \cup T^{(j')}}(x) \ge \frac{\eps}{30} \cdot \cost_{G'}(x),\]
    so adding over all $x \in G$ and choosing among the three choices randomly, we have that the total cost in expectation is at least
\[\frac{1}{3} \cdot \left(\sum_{x \in G} \frac{\eps}{30} \cdot \cost_{G'}(x)\right) = \frac{\eps}{90} \cdot \sum_{x \in G} \cost_{G'}(x) \ge \frac{\eps}{90} \cdot 0.7 \cdot \PF(\OPT) \ge \frac{\eps}{150} \cdot \PF(\OPT).\]
    Hence, for one of the three choices of $G''$, the pseudoforest cost is at least $\frac{\eps}{150} \cdot \PF(\OPT)$.
\end{proof}

\section{Coreset for Remote-Matching} \label{sec:coreset_matching}

In this section, we prove why the algorithm given in Algorithm \ref{alg:mwm_coreset} creates an $O(1)$-approximate composable coreset. First, we describe the algorithm in words.

\subsection{Algorithm Description} 

We start by running GMM on the dataset $X$ for $k$ steps, to return $k$ points $Y = \{y_1, \dots, y_k\}$. Again, let the subsets $S_1, \dots, S_k$ be a partitioning of $X,$ where $x \in X$ is in $S_i$ if $y_i$ is the closest point in $Y$ to $x$ (breaking ties arbitrarily). Note that $y_i \in S_i$ for all $i$.

To create our coreset $C$ for $X$, if $|X| \le 3k$ we simply define $C = X$. Otherwise, we start by initializing $C$ to be $Y$, so $C$ currently has size $k$. Next, for $k/2$ steps, we find any two points in $X \backslash C$ that are in the same partition piece $S_i$, and add both of them to $C$. Hence, at the end $|C| = 2k$. Note that this procedure is always doable, since we are assuming $|X| \ge 3k+1,$ which means if we have picked at most $2k$ total elements, there are $k+1$ remaining elements in $X$, of which at least $2$ must be in the same $S_i$ by the pigeonhole principle.

\subsection{Analysis}

In this subsection, we prove that the algorithm generates an $O(1)$-approximate composable coreset, by proving \Cref{thm:matching_coreset}.

\begin{proof}[Proof of Theorem \ref{thm:matching_coreset}]
    Suppose we run this algorithm for each of $m$ datasets, $X^{(1)}, \dots, X^{(m)},$ to generate coresets $C^{(1)}, \dots, C^{(m)}$. We wish to show that the optimum $k$-matching cost of $C = \bigcup_{j=1}^m C^{(j)}$ is at least $\Omega(1)$ times the optimum $k$-matching cost of $X = \bigcup_{j=1}^m X^{(j)}$.
    
    Let $Y^{(j)} = \{y_1^{(j)}, \dots, y_k^{(j)}\}$ represent the $k$ points we obtained by running GMM on $X^{(j)}$, and let $r^{(j)}$ be the maximum distance from any point in $X^{(j)} \backslash Y^{(j)}$ to its closest point in $Y^{(j)}$. Then, note that all points in $Y^{(j)}$ are pairwise separated by at least $r^{(j)}.$ Let $r = \max_{1 \le j \le m} r^{(j)}$.
    
    First, suppose that the optimum $k$-matching cost of $X$ is $M \le 5 r \cdot k$. In this case, for the $r^{(j)}$ that equals $r$, the GMM algorithm finds $k$ points that are pairwise separated by at least $r^{(j)} = r$. Since $C^{(j)} \supset Y^{(j)},$ this means that the full coreset $C$ contains $k$ points that are pairwise separated by $r$, which has $k$-matching cost at least $r \cdot \frac{k}{2}$. Hence, we have a $10$-approximate coreset.
    
    Alternatively, the optimum $k$-matching cost of $X$ is $M \ge 5r \cdot k$. Let $S_i^{(j)}$ represent the set $S_i$ for $X^{(j)}$, and suppose $W$ is an optimal set of $k$ points in $X$ with $\MWM(W) = M$. Let $W^{(j)} = W \cap X^{(j)}$. Also, let $W_i^{(j)} = W \cap S_i^{(j)}$ and $b_i^{(j)}$ be the \emph{parity} of $|W_i^{(j)}|,$ i.e., $b_i^{(j)} = 1$ if $|W_i^{(j)}|$ is odd and $b_i^{(j)} = 0$ if $|W_i^{(j)}|$ is even. In addition, let $\tilde{W}$ be the multiset of $k$ points formed by mapping each point in $W_i^{(j)}$ to $y_i^{(j)}$. In other words, $\tilde{W}$ consists of each $y_i^{(j)}$ repeated $|W_i^{(j)}|$ times. Note that since each $W_i^{(j)}$ has distance at most $r^{(j)} \le r$ from $y_i^{(j)},$ all pairwise distances change by at most $2r,$ which means the matching cost difference $|\MWM(\tilde{W}) - \MWM(W)| \le \frac{1}{2} \cdot 2r \cdot k = rk.$ Also, note that $\tilde{W}$ only consists of points of the form $y_i^{(j)}$, with the parity of the number of times $y_i^{(j)}$ appears in $\tilde{W}$ equaling $b_i^{(j)}$.
    
    Next, we create a similar set $W' \subset C$. For each $j \le m,$ define $k^{(j)} = |W^{(j)}|$. We will find a set $(W')^{(j)} \subset C^{(j)}$ of size $k^{(j)},$ such that the parity of $|(W')^{(j)} \cap S_i^{(j)}|$ equals $b_i^{(j)}$ for all $i \le k$. To do so, first note that if $|X^{(j)}| \le 3k$, then $C^{(j)} = X^{(j)},$ so we can just choose $(W')^{(j)} = W^{(j)}$. Otherwise, $|X^{(j)}| \ge 3k+1,$ and $C^{(j)}$ consists of $2k$ points. In addition, $C^{(j)} \supset Z^{(j)}$. Now, we start by including in $(W')^{(j)}$ each point $y_i^{(j)}$ such that $b_i^{(j)} = 1.$ Since $b_i^{(j)} = 1$ means that $|W_i^{(j)}|$ is odd, for any fixed $j$ the number of $b_i^{(j)} = 1$ is at most $k^{(j)}$ and has the same parity as $k^{(j)}$. Now, as long as $|(W')^{(j)}| < k^{(j)},$ this means $|C^{(j)} \backslash (W')^{(j)}| \ge k+1,$ which means there are two points in $C^{(j)} \backslash (W')^{(j)}$ that are in the same $S_i^{(j)}$, by pigeonhole principle. We can add both of them to $(W')^{(j)}$. We can keep repeating this procedure until $|(W')^{(j)}| = k^{(j)}$, and note that this never changes the parity of each $|(W')^{(j)} \cap S_i^{(j)}|$.
    
    Our set $W' \subset C$ will just be $\bigcup_{j=1}^m (W')^{(j)}.$ Note that $|W'| = \sum_{j=1}^m k^{(j)} = k$, and $|W' \cap S_i^{(j)}|$ has parity $b_i^{(j)}$, just like $W$. Hence, we can create the multiset $\tilde{W}'$ by mapping each point $w' \in W' \cap S_i^{(j)}$ to $y_i^{(j)}.$ Again, each point moves by at most $r$, so all pairwise distances change by at most $2r$, which means that $|\MWM(\tilde{W}') - \MWM(W')| \le rk.$
    
    Finally, we will see that $\MWM(\tilde{W}')= \MWM(\tilde{W})$. Note that both $\tilde{W}$ and $\tilde{W}'$ are multisets of $y_i^{(j)}$ points, each repeated an odd number of times if and only if $b_i^{(j)} = 1$. However, we saw in the proof of Lemma \ref{lem:mwm_lem_2} that $\MWM(\tilde{W})$ equals the minimum-weight matching cost of simply including each point $y_i^{(j)}$ exactly $b_i^{(j)}$ times. This is because there exists an optimal matching that keeps matching duplicate points together as long as it is possible. The same holds for $\MWM(\tilde{W}'),$ which means $\MWM(\tilde{W}) = \MWM(\tilde{W}')$.
    
    Overall, this means that $|\MWM(W')-\MWM(W)| \le |\MWM(W')-\MWM(\tilde{W}')| + |\MWM(\tilde{W}')-\MWM(\tilde{W})| + |\MWM(\tilde{W}) - \MWM(W)| \le rk + 0 + rk = 2rk$. But since we assumed that $\MWM(W) = M \ge 5rk,$ this means the $k$-matching for $C$ is at least $M - 2rk \ge \frac{M}{2}.$ Hence, we get a $2$-approximate coreset.
    
    In either case, we obtain an $O(1)$-approximate coreset, as desired.
\end{proof}

\newcommand{\etalchar}[1]{$^{#1}$}

\newpage
\appendix
\onecolumn

\section{Offline Algorithm for Remote-Pseudoforest} \label{sec:offline_pseudoforest}

In this section, we prove an alternative method of generating an $O(1)$-approximate offline remote-pseudoforest algorithm.

\begin{theorem}[Remote-Pseudoforest, Offline Algorithm] \label{thm:pseudoforest_offline}
    Given a dataset $X = \{x_1, \dots, x_n\}$ and an integer $k \le n$, Algorithm \ref{alg:pf_offline} outputs an $O(1)$-approximate set $Z$ for remote-pseudoforest. By this, we mean that $|Z| = k$ and $\PF(Z)$ is at least $\Omega(1) \times \max_{Z' \subset [n]: |Z'| = k} \PF(Z')$.
\end{theorem}

\subsection{Algorithm}

Given a dataset $X = \{x_1, \dots, x_n\},$ we recall that the goal of the remote-matching problem is to find a subset $Z = \{z_1, \dots, z_k\} \subset X$ of $k$ points, such that the pseudoforest cost of $Z$, 
\[\PF(Z) = \sum_{j=1}^k \rho(z_j, Z \backslash z_j),\]
is approximately maximized.

By scaling, we assume without loss of generality that the diameter of $X$ is $1/20$, and let $T$ be the optimum $k$-pseudoforest value. Our goal is to find a set of $k$ points that form a pseudoforest of value $\Omega(T)$.

First, we will create a series of nets. Let $\Delta$ represent the aspect ratio, which in our case, since the diameter is $1/20$, is the reciprocal of the minimum pairwise distance between two points in $X$. We will allow a runtime that is linear in $\log \Delta$, although we remark that if $n \ge 2k$, we may assume $\Delta \le O(k)$ without loss of generality. To see why, if $x, y \in X$ are two points of distance $1/20$ away from each other (which exist since $\diam(X) = 1/20$) then either at least $k$ points are of distance $1/40$ or further from $x$, or at least $k$ points are of distance $1/40$ or further from $y$. Hence, by choosing either $x$ and $k-1$ points far from $x$ or $y$ and $k-1$ points far from $y$, there exists $k$ points with minimum pseudoforest cost at least $1/40$. Therefore, if we replace every distance $\rho(x, y)$ with $\tilde{d}(x, y) = \max\left(\rho(x, y), \frac{c}{k}\right)$ for some small constant $c$ (note that $\tilde{d}$ is still a metric), every subset of $k$ points has its minimum pseudoforest cost change by at most $k \cdot \frac{c}{k} = c,$ since every pairwise distance changes by at most $\frac{c}{k}.$ Hence, if we find $k$ points with pseudoforest cost at least $\Omega(T-c)$ under $\tilde{d}$ (note that the optimum $k$-pseudoforest cost under the new metric is at least $T-c$), the same points has cost $\Omega(T-c)-c.$ For $c$ sufficiently small, this is $\Omega(T)$ since $T \ge 1/40$.

To describe our algorithm, we will first construct a tree as follows.
For each $\ell \ge 0$, let $S_\ell$ be a \emph{maximal} set of points in $X$ that are pairwise separated by $\frac{5^{-\ell}}{20}.$ Note that $S_0$ will only consist of a single point. We will create these sets in a greedy fashion starting from $S_0$, such that $S_0 \subset S_1 \subset \cdots$. We do this for $L = \lceil \log_5 \Delta \rceil$ levels, to produce $S_0 \subseteq S_1 \subseteq \dots \subseteq S_L = X$. We will create a tree based on these sets, with $S_0$ representing the root (as well as the $0^{\text{th}}$ level) and $S_\ell$ representing the nodes at the $\ell^{\text{th}}$ level.
Note that if $x \in S_\ell$ for some $\ell < L,$ $x$ is also in $S_{\ell+1}, \dots, S_L$. As a result, we represent a node $u$ in the tree as $u = x^{(\ell)},$ where $x \in X$ is the data point and $\ell$ is the level.
For any node $u = x^{\ell}$ in the tree, define $\ell(u) = \ell$ to be its level.

To construct the edges of the tree, a point $x^{(\ell+1)}$ for $x \in S_{\ell+1}$ has parent $y^{\ell}$, where $y$ is the closest point in $S_\ell$ to $x$. (If $x \in S_\ell$ also, then the parent of $x^{(\ell+1)}$ is $x^{(\ell)}$.)
Note that this distance between any node at level $\ell+1$ and its parent is at most $5^{-\ell}/20,$ or else the point could have been added to level $S_\ell,$ contradicting maximality.

Our algorithm works as follows. We will simply find a set of $k$ distinct nodes $V = \{v_1, \dots, v_k\}$ in this tree, such that no $v_i$ is an ancestor of any other $v_j$, that maximizes the value $\sum 5^{-\ell(v_i)}$.
We can write each $v_i = z_i^{(\ell_i)}$, where $\ell_i = \ell(v_i)$. Note that the $z_i$'s must be distinct points in $X$, or else the $z_i$'s would either not be distinct or one would be an ancestor of another. The algorithm outputs $Z =\{z_1, \dots, z_k\}$.

We now explain why this algorithm is efficiently implementable: we can implement it using dynamic programming. Note that the tree has total size at most $1 + L \cdot n \le O(\log \Delta \cdot n)$. Now, for each node $u$ in the tree and any $0 \le p \le k$, we define $DP[u, p]$ to be the maximum of $\sum_{v_i \in V} 5^{-\ell(v_i)},$ where $V$ is now a set of $p$ distinct nodes, with all nodes as descendants of $u$ (possibly including $u$), such that no two nodes in $V$ can be ancestors of each other. Note that $DP[u, 0] = 0$ for all nodes $u$, and $DP[u, 1] = 5^{-\ell(u)},$ since all of the descendants of $u$ have level at least $\ell(u)$. In addition, if it is impossible to pick $p$ descendants of $u$ such that no two are descendants of each other, we set $DP[u, p] = -\infty.$ Note that for any leaf node at level $L$, $DP[u, p] = -\infty$ for all $p \ge 2.$

Now, suppose that $u$ is at a level strictly smaller than $L$. Then, $u$ has $d$ children $u_1, \dots, u_d$ for some $d \ge 1$. To compute $DP[u, p]$ for $p \ge 2,$ note that we cannot use the node $u$ as it blocks all other nodes. Instead, we must use $p_1$ nodes from $u_1$ and its descendants, $p_2$ nodes from $u_2$ and its descendants, and so on, so that each $p_i \ge 0$ and $\sum_{i=1}^d p_i = p$. Hence, the dynamic programming problem we need to solve is 
\[\max\limits_{\substack{p_1, \dots, p_d \ge 0 \\ p_1 + \cdots+ p_d = p}} \left(\sum_{i=1}^{d} DP[u_i, p_i]\right).\]

Note that if $d = 1,$ the solution is simply that $DP[u, p] = DP[u_1, p]$ for all $p \ge 2.$ For $d = 2,$ this is a straightforward calculation, as we just need to maximize $DP[u_1, p_1]+DP[u_2, p-p_1]$ over all $0 \le p_1 \le p$ (with $p_2 = p-p_1$). For larger $d \ge 3,$ we can do this procedure one step at a time. Specifically, for each $1 \le d' \le d,$ we can define 
\[DP_{d'}[u, p] := \max\limits_{\substack{p_1, \dots, p_{d'} \ge 0 \\ p_1+\cdots+p_{d'} = p}} \left(\sum_{i=1}^{d'} DP[u_i,p_i]\right),\]
and note that $DP[u, p] = DP_d[u, p]$ for all $p \ge 2$. To construct $DP_{d'+1}[u, p]$ from $DP_{d'}[u, p]$, it simply equals $\max_{0 \le p' \le p} (DP_{d'}[u, p'] + DP[u_{d'+1}, p-p']).$

Hence, we can solve the dynamic program to find the maximum value. In addition, we can also find the corresponding set of nodes, by keeping track of which $p$ nodes contributed to the maximum value for each $DP[u, p]$ (and also for $DP_{d'}[u, p]$).

We provide pseudocode for the described algorithm, in Algorithm \ref{alg:pf_offline}.

\begin{algorithm}[tb]
   \caption{\textsc{PFOffline}: $O(1)$-approximate remote-pseudoforest algorithm}
   \label{alg:pf_offline}
\begin{algorithmic}[1]
    \STATE {\bfseries Input:} data $X = \{x_1, \dots, x_n\}$ (assume $\diam(X) < 1/20$), integer $k$, aspect ratio $\Delta$.
    \STATE $S_0 \leftarrow \{x_1\}$. \hspace{1cm} \textcolor{blue}{\COMMENT{$x_1$ may be replaced with an arbitrary or random point in $X$}}
    \FOR{$\ell=1$ to $L = \lceil \log_5 \Delta \rceil$}
        \STATE $S_\ell \supset S_{\ell-1}$ is a maximal set of points separated by $\frac{5^{-\ell}}{20}.$
    \ENDFOR
    \STATE \textbf{Create} tree of depth $L$, node $\ell$ levels represented by $S_\ell$, and edge connecting $x^{(\ell)}$ (for $x \in S_\ell$) and $y^{(\ell-1)}$ (for $y \in S_{\ell-1}$) if $y$ is the closest point in $S_{\ell-1}$ to $x$.
    \FOR{each leaf node $u$}
        \STATE $DP[u, 0] \leftarrow 0$, $DP[u, 1] \leftarrow 5^{-L}$, $DP[u, p] \leftarrow -\infty$ for all $2 \le p \le k$.
        \STATE $DPpoints[u, 0] = \emptyset$, $DPpoints[u, 1] = \{u\}$  \hspace{1cm} \textcolor{blue}{\COMMENT{$DPpoints[u, p]$ will keep track of the set of points to achieve value $DP[u, p]$}}
    \ENDFOR
    \FOR{each non-leaf node $u$ (starting at level $L-1$ and ending at level $0$)}
        \STATE Let $u_1, \dots, u_d$ represent the children node of $u$ in the tree.
        \FOR{$p=2$ to $k$}
            \STATE $DP[u, p] = \max\limits_{\substack{p_1, \dots, p_d \ge 0 \\ p_1 + \cdots + p_d = p}} \sum\limits_{i=1}^d DP[u_i, p_i]$.
            \IF{$DP[u, p] \neq -\infty$}
                \STATE $DPpoints[u, p] = \bigcup_{i=1}^d DPpoints[u_i, p_i],$ for the $p_i$'s chosen to maximize the above sum.
            \ENDIF
        \ENDFOR
        \STATE $DP[u, 0] \leftarrow 0$, $DP[u, 1] \leftarrow 5^{-\ell}$ (if at level $\ell$)
        \STATE $DPpoints[u, 0] \leftarrow \emptyset$, $DP[u, 1] \leftarrow u$
    \ENDFOR
    \STATE \textbf{Return} the points in $X$ corresponding to $DPpoints[x_1^{(0)}, k]$.
\end{algorithmic}
\end{algorithm}

\subsection{Analysis (Proof of Theorem \ref{thm:pseudoforest_offline})}

In this subsection, we analyze Algorithm \ref{alg:pf_offline} to show that it indeed provides an $O(1)$-approximation for the remote-pseudoforest problem. First, we prove the following lemma, which roughly states that any set of size $k$ with large pseudoforest cost induces a set $V$ of $k$ nodes with large $\sum 5^{-\ell(v)}$ value.

\begin{lemma} \label{lem:pf_analysis_1}
    Suppose there exists a set of points $Z = \{z_1, \dots, z_k\} \subset X$, such that for each $0 \le \ell \le L$, the number of indices $j \le k$ with $\rho(z_j, Z \backslash z_j) \in (5^{-(\ell+1)}, 5^{-\ell}]$ is precisely some $a_\ell$. Then, there exist $a_\ell$ distinct nodes at each level $S_\ell$ of the tree such that, across all levels, no node is an ancestor of another node.
\end{lemma}

\begin{proof}
    For some point $z_j \in X$, suppose that $\rho(z_j, Z \backslash z_j) \in (5^{-(\ell+1)}, 5^{-\ell}]$. We will \emph{map} $z_j$ to a point at level $\ell$ by starting with the node $z_j^{(L)} \in S_L$, and then finding its ancestor in level $S_\ell$ in the tree.

    It now suffices to show that we never map two points in $Z$ to the same node, or map two points in $Z$ to two distinct nodes but one is the ancestor of the other. Suppose we have some points $z_i, z_j \in Z$ which both get mapped to level $\ell$. Then, $\rho(z_i, z_j) \ge \rho(z_i, Z \backslash z_i) \ge 5^{-(\ell+1)}$.
    However, for any point $x^{(\ell+1)} \in S_{\ell+1}$ and its parent $y^{(\ell)}$ in $S_\ell$, the distance $\rho(x, y) \le 5^{-\ell}/20$.
    This means that the distance $z_j$ travels to its ancestor in level $\ell$ is at most $\frac{5^{-\ell}}{20} + \frac{5^{-(\ell+1)}}{20} + \cdots + \frac{5^{-(L-1)}}{20} \le \frac{5^{-\ell}}{16}.$
    So if $z_i, z_j$ travel to the same point in level $i$, their distance is at most $2 \cdot \frac{5^{-\ell}}{16} = \frac{5^{-\ell}}{8}$, but this is a contradiction because $\rho(z_i, z_j) \ge 5^{-(\ell+1)}$.

    Alternatively, suppose $z_i$ is mapped to level $\ell$ and $z_j$ is mapped to level $\ell' > \ell$. We still have that $\rho(z_i, z_j) \ge \rho(z_i, Z \backslash z_i) \ge 5^{-(\ell+1)}$. In addition, if the ancestor of the point $z_j$ is mapped to is the point $z_i$ is mapped to, we will still have $\rho(z_i, z_j) \le \frac{5^{-\ell}}{8}$ for the same reason.
\end{proof}

Hence, a reasonable goal is to algorithmically find $a_\ell$ distinct points in each level $\ell$, without any ancestor collision issues. But we need to make sure that if we do this, we actually output a set of $p$ points with high pseudoforest value.

\begin{lemma} \label{lem:pf_analysis_2}
    Suppose we have a subset $V$ of nodes, consisting of $a_\ell$ distinct nodes for each level $\ell$ in the tree, such that no point in the subset is an ancestor of another point.
    Then, for the corresponding points $Z$ in the set $X$, the pseudoforest value $\PF(Z) \ge \Omega(\sum a_\ell \cdot 5^{-\ell})$.
\end{lemma}

\begin{proof}
    Let's pick two nodes $x_i^{(\ell)} \in V$ at level $\ell$ and $x_j^{(\ell')} \in V$ at level $\ell' \ge \ell$. We will show that $\rho(x_i, x_j) \ge \Omega(5^{-\ell})$. Note that this also implies $\rho(x_i, x_j) \ge \Omega(5^{-\ell'})$. 

    If $\ell = \ell'$, this is straightforward, because we defined $S_\ell$ to be a maximal \emph{net} where all points in the net have distance $5^{-\ell'}/20$ from each other. So, we may assume $\ell' > \ell$.

    If $\ell' = \ell+1$, then note that the parent of $x_j$ is \emph{not} $x_i$, which means $x_j$ was closer to some point in $S_\ell$ which had distance at least $5^{-i}/20$ away from $x_i$. Therefore, $\rho(x_j, x_i) \ge 5^{-i}/40$. If $\ell' > \ell+1$, then the same is true for the ancestor of $x_j$ at level $\ell+1$.
    In other words, if $y^{(\ell+1)}$ is the ancestor of $x_j^{(\ell')},$ then $\rho(y, x_i) \ge 5^{-i}/40$.
    In addition, the distance from $x_j$ to $y$ is at most $\frac{5^{-(\ell+1)}}{20} + \frac{5^{-(\ell+2)}}{20} + \cdots + \frac{5^{-(\ell'-1)}}{20} \le \frac{5^{-\ell+1}}{16} = \frac{5^{-\ell}}{80}$. So, $\rho(x_j, x_i) \ge \frac{5^{-\ell}}{40}-\frac{5^{-\ell}}{80} = \frac{5^{-\ell}}{80}$. 
    
    Hence, if $x_i^{(\ell)} \in V$ (which means $x_i \in Z$), then $\rho(x_i, x_j) \ge \frac{1}{80} \cdot 5^{-\ell}$ for all other points $z \in Z$, so $\rho(x_i, Z \backslash x_i) \ge \frac{1}{80} \cdot 5^{-i}$. Adding over all $x_i \in Z,$ we obtain $\PF(Z) \ge \frac{1}{80} \cdot \sum a_\ell \cdot 5^{-\ell},$ as desired.
\end{proof}

To summarize, if the optimal pseudoforest cost is some $T$, achieved by some $Z_0 \subset X$ of size $k$, we can say that there are $a_\ell$ indices such that $\rho(z_j, Z \backslash z_j) \in (5^{-(\ell+1)}, 5^{-\ell}].$ Hence, $\sum a_\ell \cdot 5^{-\ell} \ge \sum_{j=1}^{k} \rho(z_j, Z \backslash z_j) \ge T$.
By Lemma \ref{lem:pf_analysis_1}, there exist $a_\ell$ distinct nodes at level $S_\ell$ for each $\ell$, such that no node is an ancestor of another node. Hence, there exist $k$ nodes $v_1, \dots, v_k$ in the tree, such that no $v_i$ is an ancestor of $v_j$ and $\sum 5^{-\ell(v_i)} \ge T$.
So, our dynamic programming algorithm finds such a set of $k$ nodes.
Finally, by Lemma \ref{lem:pf_analysis_2}, this induces a set of $k$ distinct points $Z \subset X$, with that $\PF(Z) \ge \frac{T}{80}.$ Hence, we have an $O(1)$-approximation algorithm for remote-pseudoforest, proving Theorem \ref{thm:pseudoforest_offline}.


\begin{thebibliography}{PAYLR17}

\bibitem[AAYI{\etalchar{+}}13]{abbar2013diverse}
Sofiane Abbar, Sihem Amer-Yahia, Piotr Indyk, Sepideh Mahabadi, and Kasturi~R
  Varadarajan.
\newblock Diverse near neighbor problem.
\newblock In {\em Proceedings of the twenty-ninth annual symposium on
  Computational geometry}, pages 207--214, 2013.

\bibitem[AAYIM13]{abbar2013real}
Sofiane Abbar, Sihem Amer-Yahia, Piotr Indyk, and Sepideh Mahabadi.
\newblock Real-time recommendation of diverse related articles.
\newblock In {\em Proceedings of the 22nd international conference on World
  Wide Web}, pages 1--12, 2013.

\bibitem[AFZZ15]{aghamolaei2015diversity}
Sepideh Aghamolaei, Majid Farhadi, and Hamid Zarrabi-Zadeh.
\newblock Diversity maximization via composable coresets.
\newblock In {\em CCCG}, pages 38--48, 2015.

\bibitem[AK11]{angel2011efficient}
Albert Angel and Nick Koudas.
\newblock Efficient diversity-aware search.
\newblock In {\em Proceedings of the 2011 ACM SIGMOD International Conference
  on Management of data}, pages 781--792, 2011.

\bibitem[AK17]{assadi2017randomized}
Sepehr Assadi and Sanjeev Khanna.
\newblock Randomized composable coresets for matching and vertex cover.
\newblock {\em arXiv preprint arXiv:1705.08242}, 2017.

\bibitem[AMT13]{abbassi2013diversity}
Zeinab Abbassi, Vahab~S Mirrokni, and Mayur Thakur.
\newblock {D}iversity {M}aximization {U}nder {M}atroid {C}onstraints.
\newblock In {\em Proceedings of the 19th ACM SIGKDD International Conference
  on {K}nowledge {D}iscovery and {D}ata mining}, pages 32--40, 2013.

\bibitem[BGMS16]{bhaskara2016pseudoforest}
Aditya Bhaskara, Mehrdad Ghadiri, Vahab~S. Mirrokni, and Ola Svensson.
\newblock Linear relaxations for finding diverse elements in metric spaces.
\newblock In {\em Advances in Neural Information Processing Systems}, pages
  4098--4106, 2016.

\bibitem[BLY12]{borodin2012max}
Allan Borodin, Hyun~Chul Lee, and Yuli Ye.
\newblock Max-sum diversification, monotone submodular functions and dynamic
  updates.
\newblock In {\em Proceedings of the 31st ACM SIGMOD-SIGACT-SIGAI symposium on
  Principles of Database Systems}, pages 155--166, 2012.

\bibitem[CEM18]{cevallos2018diversity}
Alfonso Cevallos, Friedrich Eisenbrand, and Sarah Morell.
\newblock Diversity maximization in doubling metrics.
\newblock {\em arXiv preprint arXiv:1809.09521}, 2018.

\bibitem[CH01]{chandra2001diversitymax}
Barun Chandra and Magn{\'{u}}s~M. Halld{\'{o}}rsson.
\newblock Approximation algorithms for dispersion problems.
\newblock {\em J. Algorithms}, 38(2):438--465, 2001.

\bibitem[CPPU16]{ceccarello2016mapreduce}
Matteo Ceccarello, Andrea Pietracaprina, Geppino Pucci, and Eli Upfal.
\newblock Mapreduce and streaming algorithms for diversity maximization in
  metric spaces of bounded doubling dimension.
\newblock {\em arXiv preprint arXiv:1605.05590}, 2016.

\bibitem[CS09]{czumaj2009sublinearmst}
Artur Czumaj and Christian Sohler.
\newblock Estimating the weight of metric minimum spanning trees in sublinear
  time.
\newblock {\em {SIAM} J. Comput.}, 39(3):904--922, 2009.

\bibitem[DP10]{drosou2010search}
Marina Drosou and Evaggelia Pitoura.
\newblock Search result diversification.
\newblock {\em ACM SIGMOD Record}, 39(1):41--47, 2010.

\bibitem[EMMZ22]{epasto2022improved}
Alessandro Epasto, Mohammad Mahdian, Vahab Mirrokni, and Peilin Zhong.
\newblock Improved sliding window algorithms for clustering and coverage via
  bucketing-based sketches.
\newblock In {\em Proceedings of the 2022 Annual ACM-SIAM Symposium on Discrete
  Algorithms (SODA)}, pages 3005--3042. SIAM, 2022.

\bibitem[EMZ19]{epasto2019scalable}
Alessandro Epasto, Vahab Mirrokni, and Morteza Zadimoghaddam.
\newblock Scalable diversity maximization via small-size composable core-sets
  (brief announcement).
\newblock In {\em The 31st ACM symposium on parallelism in algorithms and
  architectures}, pages 41--42, 2019.

\bibitem[GCGS14]{gong2014diverse}
Boqing Gong, Wei-Lun Chao, Kristen Grauman, and Fei Sha.
\newblock Diverse sequential subset selection for supervised video
  summarization.
\newblock {\em Advances in neural information processing systems}, 27, 2014.

\bibitem[Gon85]{gonzalez1985clustering}
Teofilo~F. Gonzalez.
\newblock Clustering to minimize the maximum intercluster distance.
\newblock {\em Theor. Comput. Sci.}, 38:293--306, 1985.

\bibitem[GP68]{gilbert1968steiner}
E.~N. Gilbert and H.~O. Pollak.
\newblock Steiner minimal trees.
\newblock {\em SIAM J. Appl. Math.}, 16:1--29, 1968.

\bibitem[GS09]{gollapudi2009axiomatic}
Sreenivas Gollapudi and Aneesh Sharma.
\newblock An axiomatic approach for result diversification.
\newblock In {\em Proceedings of the 18th international conference on World
  wide web}, pages 381--390, 2009.

\bibitem[HIKT99]{halldorsson1999finding}
Magn{\'u}s~M Halld{\'o}rsson, Kazuo Iwano, Naoki Katoh, and Takeshi Tokuyama.
\newblock Finding subsets maximizing minimum structures.
\newblock {\em SIAM Journal on Discrete Mathematics}, 12(3):342--359, 1999.

\bibitem[IMGR20]{indyk2020composable}
Piotr Indyk, Sepideh Mahabadi, Shayan~Oveis Gharan, and Alireza Rezaei.
\newblock Composable core-sets for determinant maximization problems via
  spectral spanners.
\newblock In {\em Proceedings of the Fourteenth Annual ACM-SIAM Symposium on
  Discrete Algorithms}, pages 1675--1694. SIAM, 2020.

\bibitem[IMMM14]{indyk2014diversity}
Piotr Indyk, Sepideh Mahabadi, Mohammad Mahdian, and Vahab~S. Mirrokni.
\newblock Composable core-sets for diversity and coverage maximization.
\newblock In Richard Hull and Martin Grohe, editors, {\em Proceedings of the
  33rd {ACM} {SIGMOD-SIGACT-SIGART} Symposium on Principles of Database
  Systems, PODS'14, Snowbird, UT, USA, June 22-27, 2014}, pages 100--108.
  {ACM}, 2014.

\bibitem[Ind04]{indyk2004emd}
Piotr Indyk.
\newblock Algorithms for dynamic geometric problems over data streams.
\newblock In L{\'{a}}szl{\'{o}} Babai, editor, {\em Proceedings of the 36th
  Annual {ACM} Symposium on Theory of Computing, Chicago, IL, USA, June 13-16,
  2004}, pages 373--380. {ACM}, 2004.

\bibitem[JSH04]{jain2004providing}
Anoop Jain, Parag Sarda, and Jayant~R Haritsa.
\newblock Providing diversity in k-nearest neighbor query results.
\newblock In {\em Pacific-Asia Conference on Knowledge Discovery and Data
  Mining}, pages 404--413. Springer, 2004.

\bibitem[LB11]{lin2011class}
Hui Lin and Jeff Bilmes.
\newblock A class of submodular functions for document summarization.
\newblock In {\em Proceedings of the 49th annual meeting of the association for
  computational linguistics: human language technologies}, pages 510--520,
  2011.

\bibitem[LBX09]{lin2009graph}
Hui Lin, Jeff Bilmes, and Shasha Xie.
\newblock Graph-based submodular selection for extractive summarization.
\newblock In {\em 2009 IEEE Workshop on Automatic Speech Recognition \&
  Understanding}, pages 381--386. IEEE, 2009.

\bibitem[MIGR19]{mahabadi2019composable}
Sepideh Mahabadi, Piotr Indyk, Shayan~Oveis Gharan, and Alireza Rezaei.
\newblock Composable core-sets for determinant maximization: A simple
  near-optimal algorithm.
\newblock In {\em International Conference on Machine Learning}, pages
  4254--4263. PMLR, 2019.

\bibitem[MZ15]{mirrokni2015randomized}
Vahab Mirrokni and Morteza Zadimoghaddam.
\newblock Randomized composable core-sets for distributed submodular
  maximization.
\newblock In {\em Proceedings of the forty-seventh annual ACM symposium on
  Theory of computing}, pages 153--162. ACM, 2015.

\bibitem[PAYLR17]{pilourdault2017motivation}
Julien Pilourdault, Sihem Amer-Yahia, Dongwon Lee, and Senjuti~Basu Roy.
\newblock Motivation-aware task assignment in crowdsourcing.
\newblock In {\em EDBT}, 2017.

\bibitem[RRT91]{ravi1991dispersion}
S.~S. Ravi, D.~J. Rosenkrantz, and G.~K. Tayi.
\newblock Facility dispersion problems: Heuristics and special cases.
\newblock {\em Algorithms and Data Structures}, 519:355--366, 1991.

\bibitem[WCO11]{welch2011search}
Michael~J Welch, Junghoo Cho, and Christopher Olston.
\newblock Search result diversity for informational queries.
\newblock In {\em Proceedings of the 20th international conference on World
  wide web}, pages 237--246, 2011.

\bibitem[YLAY09]{yu2009recommendation}
Cong Yu, Laks~VS Lakshmanan, and Sihem Amer-Yahia.
\newblock Recommendation diversification using explanations.
\newblock In {\em 2009 IEEE 25th International Conference on Data Engineering},
  pages 1299--1302. IEEE, 2009.

\bibitem[ZKL{\etalchar{+}}10]{zhou2010solving}
Tao Zhou, Zolt{\'a}n Kuscsik, Jian-Guo Liu, Mat{\'u}{\v{s}} Medo,
  Joseph~Rushton Wakeling, and Yi-Cheng Zhang.
\newblock Solving the apparent diversity-accuracy dilemma of recommender
  systems.
\newblock {\em Proceedings of the National Academy of Sciences},
  107(10):4511--4515, 2010.

\bibitem[ZMKL05]{ziegler2005improving}
Cai-Nicolas Ziegler, Sean~M McNee, Joseph~A Konstan, and Georg Lausen.
\newblock Improving recommendation lists through topic diversification.
\newblock In {\em Proceedings of the 14th international conference on World
  Wide Web}, pages 22--32, 2005.

\end{thebibliography}
\end{document}